\documentclass[oribibl]{llncs}
\usepackage{makeidx}  

\usepackage{amssymb}
\usepackage{graphicx, subfigure}
\usepackage{wrapfig}
\usepackage[ruled,vlined,linesnumbered]{algorithm2e}
\usepackage{cutwin}
\usepackage{caption}
\usepackage{float} 

\graphicspath{{./graphics/}}

\usepackage{hyperref}
\hypersetup{
colorlinks=true, 
linkcolor=blue, 
citecolor=green, 
urlcolor=webbrown, 
pdftitle={CMST}, 
breaklinks=true, bookmarks=true,bookmarksnumbered,
pdfauthor={Alina Shaikhet}, 
pdfsubject={}, 
pdfkeywords={}, 
pdfcreator={pdfLaTeX}, 
pdfproducer={LaTeX with hyperref and ClassicThesis} 
}
\begin{document}
\frontmatter          
\pagestyle{headings}  
\mainmatter              
\title{Essential Constraints of Edge-Constrained Proximity Graphs\thanks{This work was partially supported by NSERC.}
\thanks{A preliminary version of this paper appeared in the Proceedings of 27th International Workshop, IWOCA 2016, Helsinki, Finland.}}
\titlerunning{Edge-Constrained Proximity Graphs}  
%
\author{Prosenjit Bose\inst{1} \and Jean-Lou De Carufel\inst{2} \and Alina Shaikhet\inst{1} \and Michiel Smid\inst{1}}
\authorrunning{P. Bose, J.-L. De Carufel, A. Shaikhet and M. Smid} 

\institute{School of Computer Science, Carleton University, Ottawa, Canada,\\
\email{\{jit, michiel\}@scs.carleton.ca, alina.shaikhet@carleton.ca},\\
\and
School of Electrical Engineering and Computer Science, U. of Ottawa, Canada,\\
\email{jdecaruf@uottawa.ca}}

\maketitle              

\begin{abstract}
Given a plane forest $F = (V, E)$ of $|V| = n$ points, we find the minimum set $S \subseteq E$ of edges such that the edge-constrained minimum spanning tree over the set $V$ of vertices and the set $S$ of constraints contains $F$. We present an $O(n \log n )$-time algorithm that solves this problem.
We generalize this to other proximity graphs in the constraint setting, such as the \emph{relative neighbourhood graph}, \emph{Gabriel graph}, $\beta$\emph{-skeleton} and \emph{Delaunay triangulation}.


We present an algorithm that identifies the minimum set $S\subseteq E$ of edges of a given plane graph $I=(V,E)$ such that $I \subseteq CG_\beta(V, S)$ for $1 \leq \beta \leq 2$, where $CG_\beta(V, S)$ is the constraint $\beta$-skeleton over the set $V$ of vertices and the set $S$ of constraints. The running time of our algorithm is $O(n)$, provided that the constrained Delaunay triangulation of $I$ is given. 
\keywords{proximity graphs, constraints, visibility, MST, Delaunay, $\beta$-skeletons}
\end{abstract}
\section{Introduction}
\label{sec:introduction}

This paper was inspired by topics in geometric compression. In particular, Devillers et~al.~\cite{DBLP:journals/ijcga/DevillersEGHRS03} investigate how to compute the minimum set $S \subseteq E$ of a given plane triangulation $T = (V,E)$, such that $T$ is a constrained Delaunay triangulation $(DT)$ of the graph $(V,S)$. They show that $S$ and $V$ is the only information that needs to be stored. The graph $T$ can be successfully reconstructed from $S$ and $V$. Experiments on real data sets (such as terrain models and meshes) show that the size of $S$ is less than $3.4\%$ of the total number of edges of $T$, which yields an effective compression of the triangulation.

Our goal is to broaden this research and investigate geometric compression of other neighbourhood graphs. We study minimum spanning trees, relative neighbourhood graphs, Gabriel graphs and $\beta$-skeletons for $1 \leq \beta \leq 2$. We give a definition of each of those graphs in the constraint setting (refer to Sect.~\ref{sec:basic_definitions}). 

Minimum spanning trees $(MST)$ have been studied for over a century and have numerous applications. We study the problem of finding the minimum set $S$ of constraint edges in a given plane forest $F = (V, E)$ such that the edge-constrained $MST$ over the set $V$ of vertices and the set $S$ of constraints contains~$F$. If $F$ is a plane tree then the edge-constrained $MST$ over $(V,S)$ is equal to $F$. We give an $O(n \log n )$-time algorithm that solves this problem.

Gabriel graphs $(GG)$ were introduced by Gabriel and Sokal in~\cite{GabrielSokal1969}. 
Toussaint introduced the notion of relative neighbourhood graphs ($RNG$) in his research on pattern recognition~\cite{DBLP:journals/pr/Toussaint80}. Both graphs were studied extensively.

Jaromczyk and Kowaluk showed that $RNG$ of a set $V$ of points can be constructed from the Delaunay triangulation of $V$ in time $O(n \alpha(n, n))$, where $\alpha(\cdot)$ is the inverse of the Ackerman function~\cite{Jaromczyk:1987:NRN:41958.41983}. These two authors, together with Yao, improved the running time of their algorithm to linear~\cite{Jaromczyk:1989}. They achieved it by applying a static variant of the \emph{Union-Find} data structure. They also generalized their algorithm to construct the $\beta$-skeleton $(G_{\beta})$ for $1 \leq \beta \leq 2$ in linear time from the Delaunay triangulation of $V$ under the $L_p$-metric, $1 < p < \infty$. We provide the definition of $\beta$-skeleton in Sect.~\ref{sec:basic_definitions}. For now, note that the $1$-skeleton corresponds to the Gabriel graph and the $2$-skeleton corresponds to the relative neighbourhood graph. In this paper, we use two geometric structures: elimination path and elimination forest, introduced by Jaromczyk and Kowaluk~\cite{Jaromczyk:1987:NRN:41958.41983}.

Neighbourhood graphs are known to form a nested hierarchy, one of the first versions of which was established by Toussaint~\cite{DBLP:journals/pr/Toussaint80}: for any $1 \leq \beta \leq 2$, $ MST \subseteq RNG \subseteq G_{\beta} \subseteq GG  \subseteq DT$.
We show that the neighbourhood graphs in the constraint setting form the same hierarchy. Moreover, we show that the minimum set of constraints required to reconstruct a given plane graph (as a part of each of those neighbourhood graphs) form an inverse hierarchy.

In Sect.~\ref{sec:basic_definitions}, we present notations and definitions. In Sect.~\ref{sec:CMST algorithm}, we give some observations concerning constrained $MST$, show worst-case examples and present an $O(n \log n )$-time algorithm that identifies the minimum set $S \subseteq E$ of constraint edges given a plane forest $F = (V, E)$ such that the edge-constrained $MST$ over the set $V$ of vertices and the set $S$ of constraints contains $F$. In Sect.~\ref{sec:CGG} we address the special case of constrained $\beta$-skeletons - constrained Gabriel graph ($1$-skeleton). Although the algorithm given in Sect.~\ref{sec:beta-skeleton} can be successfully applied to constrained Gabriel graphs, the algorithm presented in Sect.~\ref{sec:CGG} is significantly simpler and requires only local information about an edge in question. Section~\ref{sec:beta-skeleton} presents an algorithm that identifies the minimum set $S$ of edges of a given plane graph $I = (V,E)$ such that $I \subseteq CG_\beta(V, S)$ for $1 \leq \beta \leq 2$, where $CG_\beta(V, S)$ is a constrained $\beta$-skeleton on the set $V$ of vertices and set $S$ of constraints. The hierarchy of the constrained neighbourhood graphs together with the hierarchy of the minimum sets of constraints are given in Sect.~\ref{sec:Hierarchy}.

\section{Basic Definitions}
\label{sec:basic_definitions}

Let $V$ be a set of $n$ points in the plane and $I = (V, E)$ be a plane graph representing the constraints. Each pair of points $u$, $v \in V$ is associated with a neighbourhood defined by some property $P(u,v,I)$ depending on the proximity graph under consideration. An edge-constrained neighbourhood graph $G_P(I)$ defined by the property $P$ is a graph with a set $V$ of vertices and a set $E_P$ of edges such that $\overline{uv} \in E_P$ if and only if $\overline{uv} \in E$ or $\overline{uv}$ satisfies $P(u,v,I)$. 

For clarity and to distinguish between different types of input graphs, if $I$ is a forest, we will denote $I$ by $F = (V, E)$, to emphasize its properties. 

In this paper, we assume that the points in $V$ are in general position (no three points are collinear and no four points are co-circular). 

Two vertices $u$ and $v$ are \emph{visible to each other with respect to $E$} provided that $\overline{uv} \in E$ or the line segment $\overline{uv}$ does not intersect the interior of any edge of $E$. For the following definitions, let $I = (V, E)$ be a plane graph.

\begin{definition}[Visibility Graph of I]
\label{def:VG(I)}
The \emph{visibility graph of $I$} is the graph $VG(I) = (V, E')$ such that $E' = \{ (u, v): u,v \in V$, $u$ and $v$ are \emph{visible} to each other with respect to $E\}$. It is a simple and unweighted graph.
\end{definition}

In Def.~\ref{def:VG(I)}, we may think of $I$ as of a set of obstacles. The nodes of $VG(I)$ are the vertices of $I$, and there is an edge between vertices $u$ and $v$ if they can \emph{see} each other, that is, if the line segment $\overline{uv}$ does not intersect the interior of any obstacle in $I$. 
We say that the endpoints of an obstacle edge see each other. Hence, the obstacle edges form a subset of the edges of $VG(I)$, and thus~$I \subseteq VG(I)$.

\begin{definition}[Euclidean Visibility Graph of I]
\label{def:EVG(I)}
The \emph{Euclidean visibility graph} $EVG(I)$ is the visibility graph of $I$, where each edge $\overline{uv}$ $(u,v \in V)$ is assigned weight $w(u,v)$ that is equal to the Euclidean distance between $u$ and $v$.
\end{definition}

\begin{definition}[Constrained Visibility Graph of I]
\label{def:CVG(I)}
The \emph{constrained visibility graph} $CVG(I)$ is the visibility graph of $I$, where each edge of $I$ is assigned weight $0$ and every other edge has weight equal to its Euclidean length.
\end{definition}

We use the notation $MST(G)$ to refer to a minimum spanning tree of the graph $G$. We assume that each edge of $G$ has weight equal to its Euclidean length, unless the edge was specifically assigned the weight $0$ by our algorithm. If none of the edges of $G$ are assigned the weight $0$ then $MST(G)$ is a Euclidean minimum spanning tree of $G$.

\begin{definition}[Constrained Minimum Spanning Tree of F]
\label{def:CMST(F)}
Given a plane forest $F = (V, E)$, the \emph{constrained minimum spanning tree} $CMST(F)$ is the minimum spanning tree of $CVG(F)$. 
\end{definition}

We assume that all the distances between any two vertices of $V$ are distinct, otherwise, any ties can be broken using lexicographic ordering. This assumption implies that there is a unique $MST$ and a unique $CMST$.

Since each edge of a plane forest $F$ has weight zero in $CVG(F)$, by running Kruskal on $CVG(F)$, we get $F \subseteq CMST(F)$. Notice also that if $F$ is a plane tree, then $F = CMST(F)$.

\begin{definition}[Locally Delaunay criterion]
\label{def:locallyDelaunay}
Let $G$ be a triangulation, $\overline{v_1 v_2}$ be an edge in $G$ (but not an edge of the convex hull of $G$), and $\triangle(v_1, v_2, v_3)$ and $\triangle(v_1, v_2, v_4)$ be the triangles adjacent to $\overline{v_1 v_2}$ in $G$. We say that $\overline{v_1 v_2}$ is a \emph{locally Delaunay edge} if the circle through $\{v_1, v_2, v_3 \}$ does not contain $v_4$ or equivalently if the circle through $\{v_1, v_2, v_4 \}$ does not contain $v_3$. Every edge of the convex hull of $G$ is also considered to be \emph{locally Delaunay}~\cite{DBLP:journals/ijcga/DevillersEGHRS03}.
\end{definition}

\begin{definition}[Constrained Delaunay Triangulation of I]
\label{def:CDT(F)}
The \emph{constrained Delaunay triangulation} $CDT(I)$ is the unique triangulation of $V$ such that each edge is either in $E$ or locally Delaunay. It follows that $I \subseteq CDT(I)$.
\end{definition}

This definition is equivalent to the classical definition used for example by Chew in~\cite{DBLP:journals/algorithmica/Chew89}:
$CDT(I)$ is the unique triangulation of $V$ such that each edge $e$ is either in $E$ or there exists a circle $C$ with the following properties:
\begin{enumerate}
\item The endpoints of edge $e$ are on the boundary of $C$.
\item Any vertex of $I$ in the interior of $C$ is not visible~to~at~least~one~endpoint~of~$e$.
\end{enumerate}

The equivalence between the two definitions was shown by Lee and Lin~\cite{DBLP:journals/dcg/LeeL86}.

When considering edge weights of $CDT(I)$, we assume that the weight of each edge is equal to the Euclidean distance between the endpoints of this edge.

The relative neighbourhood graph ($RNG$) was introduced by Toussaint in 1980 as a way of defining a structure from a set of points that would match human perceptions of the shape of the set~\cite{DBLP:journals/pr/Toussaint80}. An $RNG$ is an undirected graph defined on a set of points in the Euclidean plane by connecting two points $u$ and $v$ by an edge if there does not exist a third point $p$ that is closer to both $u$ and $v$ than they are to each other. Formally, we can define $RNG$ through the concept of a $lune$. Let $D(x,r)$ denote an open disk centered at $x$ with radius $r$, i.e., $D(x,r) = \{y: dist(x,y) < r\}$. Let $L_{u,v} = D(u,dist(u,v)) \cap D(v,dist(u,v))$; $L_{u,v}$ is called a $lune$.

\begin{definition}[Relative Neighbourhood graph of V]
\label{def:RNG(V)}
Given a set $V$ of points, the \emph{Relative Neighbourhood graph of $V$},  $RNG(V)$, is the graph with vertex set $V$ and the edges of $RNG(V)$ are defined as follows:
$\overline{uv}$ is an edge if and only if $L_{u,v} \cap V = \emptyset$.
\end{definition}

\begin{definition}[Constrained Relative Neighbourhood graph of I]
\label{def:CRNG(I)}
The \emph{constrained Relative Neighbourhood graph}, $CRNG(I)$, is defined as the graph with vertices $V$ and the set $E'$ of edges such that each edge $e = \overline{uv}$ is either in $E$ or, $u$ and $v$ are visible to each other and $L_{u,v}$ does not contain points in $V$ visible from both $u$ and $v$. It follows that $I \subseteq CRNG(I)$.
\end{definition}

Gabriel graphs were introduced by Gabriel and Sokal in the context of geographic variation analysis~\cite{GabrielSokal1969}. The Gabriel graph of a set $V$ of points in the Euclidean plane expresses the notion of proximity of those points. It is the graph with vertex set $V$ in which any points $u$ and $v$ of $V$ are connected by an edge if $u \neq v$ and the closed disk with $\overline{u v}$ as a diameter contains no other point of $V$. 

\begin{definition}[Locally Gabriel criterion]
\label{def:locallyGabriel}
The edge $\overline{u v}$ of the plane graph $G = (V, E)$ is said to be \emph{locally Gabriel} if the vertices $u$ and $v$ are visible to each other and the circle with $\overline{u v}$ as a diameter does not contain any points in $V$ which are visible from both $u$ and $v$. Refer to Fig.~\ref{fig:MSTinCDT}.
\end{definition}

\begin{wrapfigure}{r}{0.32\textwidth}
\vspace{-20pt}
\centering
\includegraphics[width=0.20\textwidth]{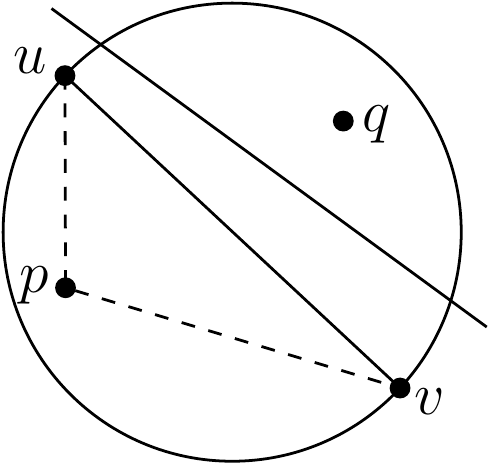}
\caption{Removal of $p$ makes $\overline{uv}$ locally Gabriel.}
\label{fig:MSTinCDT}
\end{wrapfigure}

\begin{definition}[Constrained Gabriel graph of I]
\label{def:CGG(I)}
The \emph{constrained Gabriel graph} $CGG(I)$ is defined as the graph with vertices $V$ and the set $E'$ of edges such that each edge is either in $E$ or locally Gabriel. It follows that $I \subseteq CGG(I)$.
\end{definition}

Relative neighbourhood and Gabriel graphs are special cases of a parametrized family of neighbourhood graphs called $\beta$-skeletons (defined by Kirkpatrick and Radke in~\cite{Kirkpatrick1985217}). The neighbourhood $U_{u,v}(\beta)$ is defined for any fixed $\beta$ ($1 \leq \beta < \infty$) as the intersection of two disks (refer to Fig.~\ref{fig:beta_eye}):
$U_{u,v}(\beta) = D \left( \Big( 1- \frac{\beta}{2}\Big)u + \frac{\beta}{2} v, \frac{\beta}{2} dist(u,v) \right) \cap $\\
\-\hspace{43pt} $D \left( \Big( 1- \frac{\beta}{2}\Big)v + \frac{\beta}{2} u, \frac{\beta}{2} dist(u,v) \right)$

\begin{definition}[(lune-based) $\beta$-skeleton of V]
\label{def:B-Skeleton(V)}
Given a set $V$ of points in the plane, the \emph{(lune-based) $\beta$-skeleton of V}, denoted $G_\beta(V)$ is the graph with vertex set $V$ and the edges of $G_\beta(V)$ are defined as follows: $\overline{uv}$ is an edge if and only if $U_{u,v}(\beta) \cap V = \emptyset$.
\end{definition}

Notice that $RNG(V)$ is a $\beta$-skeleton of $V$ for $\beta = 2$; namely $RNG(V) = G_2(V)$. Similarly, $GG(V) = G_1(V)$.

\begin{definition}[Constrained $\beta$-skeleton of I]
\label{def:Constrained_B-Skeleton(V)}
The \emph{constrained $\beta$-skeleton of I}, $CG_\beta(I)$ is the graph with vertex set $V$ and edge set $E'$ defined as follows: $e= \overline{uv} \in E'$ if and only if $e \in E$ or $u$ and $v$ are visible to each other and $U_{u,v}(\beta)$ does not contain points in $V$ which are visible from both $u$ and $v$. 
\end{definition}

\section{CMST Algorithm}
\label{sec:CMST algorithm}

\textbf{Problem 1:} Let a plane forest $F = (V, E)$ with $|V| = n$ points be given. Find the minimum set $S \subseteq E$ of edges such that $F \subseteq CMST(V,S)$. 

In other words, we want to find the smallest subset $S$ of edges of $F$ such that $CMST(F)$ is \textit{equal} to $CMST(V,S)$, although the weights of the two trees may be different. Recall, that $CMST(F) = CMST(V,E)$ is the minimum spanning tree of the weighted graph $CVG(V,E)$ where each edge of $E$ is assigned weight $0$, and every other edge is assigned a weight equal to its Euclidean length. Notice, that if $F$ is a tree then $F = CMST(F)$. If $F$ is disconnected then for every edge $e$ of $CMST(F)$ such that $e \notin E$ there exists a cut in $CVG(V,E)$ such that $e$ belongs to the cut and its weight is the smallest among other edges that cross the cut (this is true for the same cut in $EVG(V,E)$).
 
Let us begin by considering an example. We are given a tree $F = (\{v_1, v_2, v_3\},\\ \{ \overline{v_1 v_2}, \overline{v_2 v_3}\})$ (refer to Fig.~\ref{fig:example1}). Figure~\ref{fig:example11} shows $CDT(F)$. Observe that 
$CDT(F) = DT(\{v_1, v_2, v_3\})$. In other words, $CDT(F) = CDT(\{v_1, v_2, v_3\}, \emptyset)$ and thus no constraints are required to construct $CDT(F)$. However, this is not the case with $CMST(F)$. 
Obviously $MST(CDT(F)) \neq CMST(F)$ (refer to Fig.~\ref{fig:example111}), because $F \nsubseteq MST(CDT(F))$. We need to identify the minimum set $S \subseteq F$ of edges such that $F = CMST(V,S)$. In this example $S = \{ \overline{v_1 v_2}, \overline{v_2 v_3}\}$.

\begin{figure}[h]
\centering
\subfigure[]{%
		\label{fig:example1}%
		\includegraphics[width=0.26\textwidth]{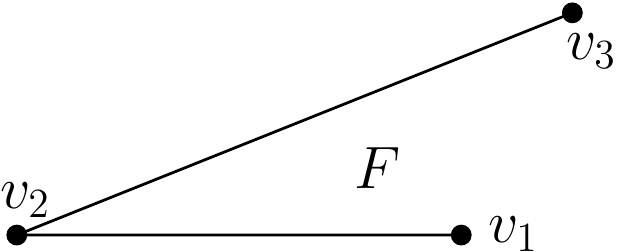}}
\hspace{0.04\textwidth}
\subfigure[]{%
		\label{fig:example11}%
		\includegraphics[width=0.26\textwidth]{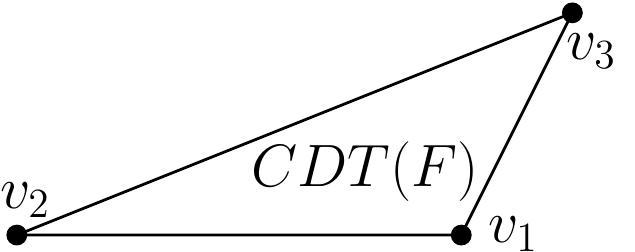}}
\hspace{0.04\textwidth}
\subfigure[]{%
		\label{fig:example111}%
		\includegraphics[width=0.26\textwidth]{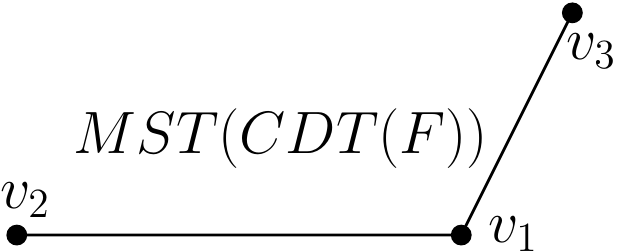}}
\caption{Example showing relationship between input graph, its $CDT$ and $MST(CDT)$. \textbf{(a)} Input graph $F$. \textbf{(b)} No constraints are required to construct $CDT(F)$. \textbf{(c)} $F \neq MST(CDT(F))$. We need two constraints $S = \{ \overline{v_1 v_2}, \overline{v_2 v_3}\}$.}
\label{fig:example}
\end{figure}

A first idea is to construct an $MST$ of $EVG(V,\emptyset)$. Every edge of $F$ that is not part of $MST(EVG(V,\emptyset))$ should be forced to appear in $CMST(F)$. If we do this by adding each such edge of $F$ to $S$ (recall that every edge in $S$ has weight $0$) then, unfortunately, some edges of $F$, that were part of $MST(EVG(V,\emptyset))$, will no longer be part of the $MST$ of the updated graph. A second approach is to start with $MST(EVG(V,\emptyset))$ and \textit{eliminate} every edge that is not part of $F$ and does not connect two disconnected components of $F$. Each such edge  $e \in MST(EVG(V,\emptyset))$ creates a cycle $c_e$ in $CMST(F) \cup \{ e \}$ and we have that $c_e \subseteq EVG(V,\emptyset)$. If $e$ becomes the heaviest edge of $c_e$ then it will no longer be part of the $MST$. Thus, we add to $S$ every edge of $c_e \cap E$ that is heavier than $e$. Although this approach gives us a set $S$ such that $F \subseteq CMST(V,S)$, the set $S$ of edges with weight $0$ may not be minimal. Consider the example of Fig.~\ref{fig:optimality_counterExample}. We are given a tree $F = (\{v_1, v_2, v_3, v_4\}, \{ \overline{v_1 v_2}, \overline{v_2 v_3}, \overline{v_3 v_4}\})$ (refer to Fig.~\ref{fig:optimality_counterExample_F}). Every edge on the path from $v_1$ to $v_4$ in $F$ is heavier than $\overline{v_1 v_4}$ - an edge of $MST(EVG(V,\emptyset)) \setminus F$. In order to eliminate $\overline{v_1 v_4}$ from the $MST$ we assign the weight $0$ to all the edges of the path $c_{\overline{v_1 v_4}} \setminus \overline{v_1 v_4}$, i.e. $S = \{ \overline{v_1 v_2}, \overline{v_2 v_3}, \overline{v_3 v_4}\}$. However, it is sufficient to assign weight $0$ only to the edge $\overline{v_2 v_3}$. In this case, $CMST(V,\{\overline{v_2 v_3}\}) = F$. 

\begin{figure}[h]
\centering
\subfigure[]{%
		\label{fig:optimality_counterExample_F}%
		\includegraphics[width=0.30\textwidth]{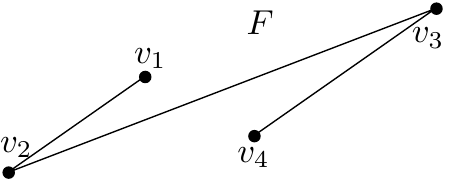}}
\subfigure[]{%
		\label{fig:optimality_counterExample_MST}%
		\includegraphics[width=0.30\textwidth]{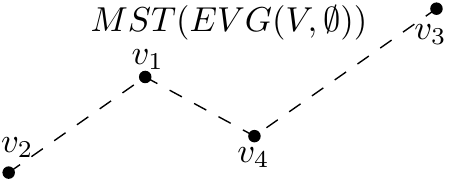}}
\subfigure[]{%
		\label{fig:optimality_counterExample_cycle}%
		\includegraphics[width=0.30\textwidth]{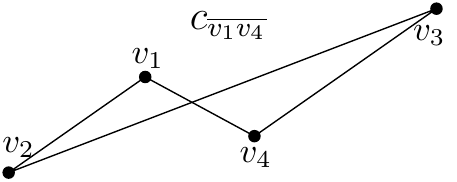}}	
\caption{Counterexample to the optimality of $S = \{ \overline{v_1 v_2}, \overline{v_2 v_3}, \overline{v_3 v_4}\}$. The set $S' = \{ \overline{v_2 v_3}\}$ is optimal. \textbf{(a)} Input graph $F = (V,E)$. \textbf{(b)} $MST$ of $V$. \textbf{(c)} Cycle $c_{\overline{v_1 v_4}}$ of the graph $F \cup \{ \overline{v_1 v_4} \}$. Every edge on the path from $v_1$ to $v_4$ in $F$ is heavier than $\overline{v_1 v_4}$.}
\label{fig:optimality_counterExample}
\end{figure}

As we will see later this second approach is correct when applied to the $MST$ of a different graph. Instead of considering edges of $MST(EVG(V,\emptyset))$ we apply our idea to $MST(EVG(F))$. Notice that $EVG(F)$ does not have edges that intersect edges of $F$, and thus we will not encounter cases similar to the example of Fig.~\ref{fig:optimality_counterExample}. Now it may look like we will be missing important information by considering only a subset of $VG(V, \emptyset)$. Can we guarantee that $CMST(V,S)$ will not contain edges that intersect edges of $F \setminus S$? To answer this question, we prove the following statement: $CMST(V, S) \subseteq VG(F)$ (Lemma~\ref{lem:CMSTinVG(T)}). The basic algorithm for constructing $S$ is given below.  We prove its optimality by showing minimality of $S$ (Lemma~\ref{lem:minimal&minimum}). Later, we present an efficient implementation of this algorithm.
 
By $CDT^\circ (F)$ we denote $CDT(F)$ where each edge of $F$ is assigned weight~$0$.

\begin{algorithm}[H]
\caption{Minimum set of constraints for $CMST$}\label{alg_CMST}
\KwIn{plane forest $F = (V, E)$}
\KwOut{minimum set $S \subseteq E$ of constraints such that $F \subseteq CMST(V,S)$}
\BlankLine
Construct $T' = MST(EVG(F))$
\tcp*[r]{we show $T' = MST(CDT(F))$}
Construct $CMST(F)$
\tcp*[r]{$CMST(F) = MST(CDT^\circ (F))$}
Initialize $S = \emptyset$\;
\ForEach{$e' \in T'$}{
	\If{$CMST(F) \cup \{e'\}$ creates a cycle $c_{e'}$}{
		\ForEach{$e \in c_{e'} \cap E$}{
			\If{$w(e) > w(e')$}{set $S \leftarrow S \cup \{ e \}$}
		}
	}
}
\end{algorithm}

We show the correctness of Algorithm~\ref{alg_CMST} by proving Lemmas~\ref{lem:CMSTinVG(T)} -- \ref{lem:MST(EVG(F))inCDT(F)}. We start by observing an interesting property of the edges of $F$ that were \textbf{not} added to $S$ during the execution of the algorithm. 

\begin{property}
\label{property1}
Let $S$ be the output of Algorithm~\ref{alg_CMST} on the input plane forest $F = (V, E)$. Let $T' = MST(EVG(F))$. If $e = \overline{uv} \in F$ and $e \notin S$ then $e \in T'$. 
\end{property}

\begin{proof}
Assume to the contrary that $e \notin T'$. If we add $e$ to $T'$, this creates a cycle $c_e$. Notice that $e$ is the longest edge among the edges of $c_e$. It is given that $e \in F$ and thus $e \in CMST(F)$ by Def.~\ref{def:CMST(F)}. Let us look at the cut $U$, $V \setminus U$ in $CMST(F)$ such that $u \in U$, $v \in V \setminus U$ (refer to Fig.~\ref{fig:cut}). In $T'$ there is a path from $u$ to $v$ that does not contain $e$. This path is $c_e \setminus \{e\}$. There exists an edge $e' \in T'$ such that $e' \in c_e$ and $e'$ belongs to the cut $U$, $V \setminus U$. Notice that $e' \notin CMST(F)$, otherwise $CMST(F)$ has a cycle. Since $e$ is the longest edge of the cycle $c_e$ then $|e| > |e'|$. Because $e' \in T'$ (and $e' \notin CMST(F)$), Algorithm~\ref{alg_CMST} will consider $e'$ in steps $4$ -- $8$. Since $e'$ belongs to the cut $U$, $V \setminus U$ in $CMST(F)$, the cycle $c_{e'}$ in $CMST(F) \cup \{ e' \}$ is formed. Notice that $e \in c_{e'}$ because $e$ and $e'$ belong to the cut $U$, $V \setminus U$. Since $|e| > |e'|$, the edge $e$ is added to $S$. This is a contradiction to $e \notin S$ and thus $e \in T'$.
\qed
\end{proof}


\begin{figure}[h]
  \begin{minipage}[c]{0.5\textwidth}
    \includegraphics[width=0.84\textwidth]{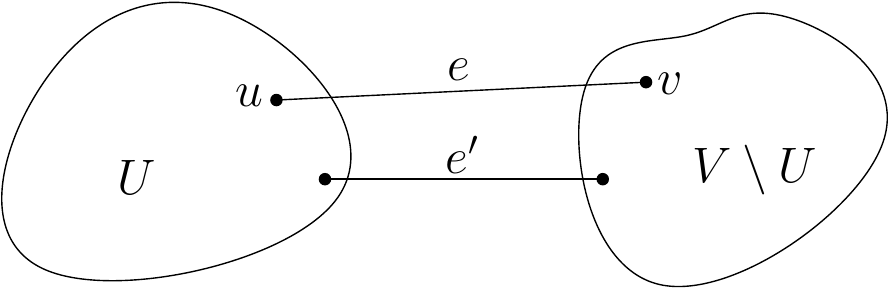}
  \end{minipage}\hfill
  \begin{minipage}[c]{0.5\textwidth}
    \caption{The cut $U$, $V \setminus U$ in $CMST(F)$ s.t. $u \in U$, $v \in V \setminus U$.} 
    \label{fig:cut}
  \end{minipage}
\end{figure}

\begin{lemma}
\label{lem:CMSTinVG(T)}
Let $S$ be the output of Algorithm~\ref{alg_CMST} on the input plane forest $F = (V, E)$. We have $CMST(V, S) \subseteq VG(F)$. 
\end{lemma}

\begin{proof}
Let $e^* = \overline{ab}$ be an arbitrary edge of $CMST(V,S)$. Assume to the contrary that $e^* \notin VG(F)$ (and hence $e^* \notin F$ and $e^* \notin S$). Thus there exists an edge of $F$ that intersects $e^*$. Notice, that this edge cannot be in $S$. 

Let $k$ ($1 \leq k \leq n$) be the number of edges of $F$ that intersect $e^*$. Let $e_i = \overline{c_i d_i} \in F$ be the edge that intersects $e^*$ at point $x_i$, where $i$ ($0 \leq i < k$) represents an ordering of edges $e_i$ according to the length $|\overline{ax_i}|$. In other words, the intersection point between $e^*$ and $e_0$ is the closest to $a$ among other edges of $F$ that intersect $e^*$. Refer to Fig.~\ref{fig:star_bars}.

\begin{figure}[h]
  \begin{minipage}[c]{0.5\textwidth}
    \includegraphics[width=0.84\textwidth]{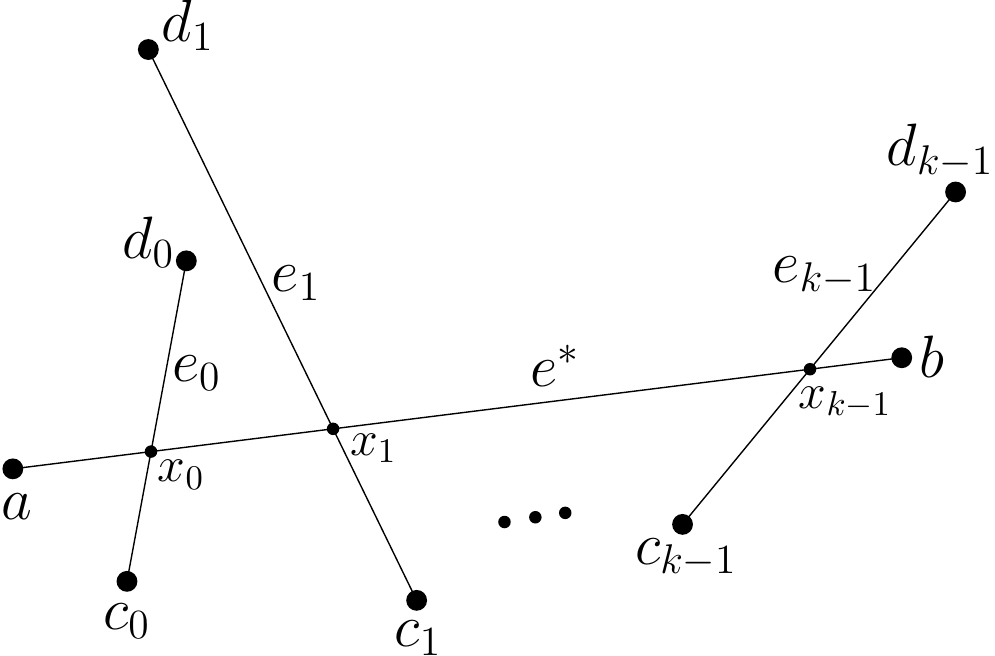}
  \end{minipage}\hfill
  \begin{minipage}[c]{0.5\textwidth}
    \caption{Intersection between $e^* \in CMST(V,S)$ and $k$ edges of $F$. Notice that the points $x_0, \ldots, x_{k-1}$ do not belong to $V$.} \label{fig:star_bars}
  \end{minipage}
\end{figure}

  
We prove this lemma in three steps. First we derive some properties of $e_i$. Then we show that both endpoints of $e^*$ are outside the disk with $e_i$ as a diameter for every $0 \leq i < k$. We finalize the proof by showing that $e^* \notin CGG(V,S)$ and thus by Lemma~\ref{lem:Hierarchy} (establishing that $CMST(V,S) \subseteq CGG(V,S)$, refer to Sect.~\ref{sec:Hierarchy}) we have $e^* \notin CMST(V,S)$. This contradicts the definition of $e^*$ which leads to the conclusion that $e^* \in VG(F)$ (meaning that the intersection between $e^*$ and an edge of $F$ is not possible). 

\textbf{Step 1:} Since $e_i \in F$ and $e_i \notin S$ then by Property~\ref{property1}: $e_i \in MST(EVG(F))$. Let $D_{e_i}$ be the disk with $e_i$ as a diameter. There does not exist a point of $V$ inside $D_{e_i}$ that is visible to both $c_i$ and $d_i$. Otherwise, if such a point $v$ exists, then the cycle $\{ \overline{vc_i}, e_i, \overline{d_i v} \}$ is part of $VG(F)$ with $e_i$ being the longest edge of this cycle. This contradicts $e_i \in MST(EVG(F))$.

\textbf{Step 2:} Let us first consider $e_0$ and show that $a \notin D_{e_0}$. Refer to Fig.~\ref{fig:star_zerobar}. 

\begin{figure}[h]
\centering
\subfigure[In bold we highlight line segments that cannot be intersected by edges of $F$.]{
		\label{fig:star_zerobar}
		\includegraphics[width=0.38\textwidth]{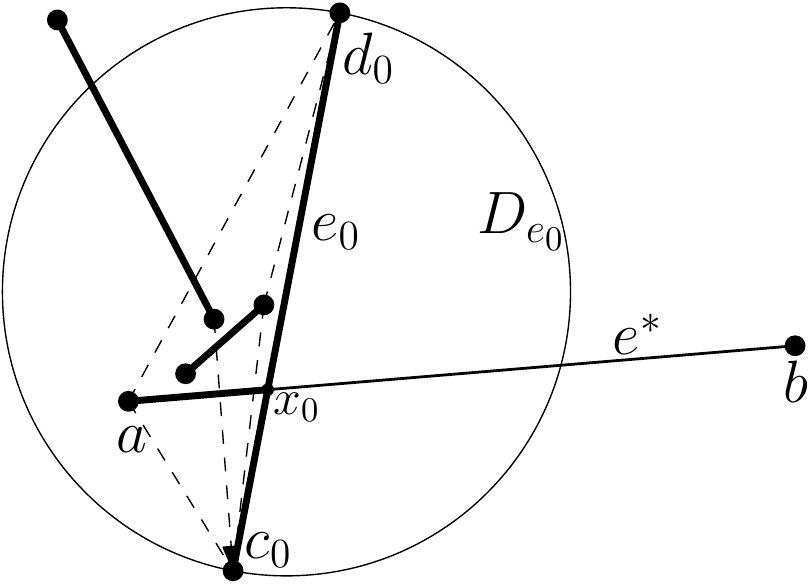}}
\hspace{0.05\textwidth}
\subfigure[Notice that $\overline{a x_0}$, $\overline{x_0 x_1}$, $e_0$ and $e_1$ cannot be intersected by edges of $F$.]{
		\label{fig:star_secondbar}
		\includegraphics[width=0.34\textwidth]{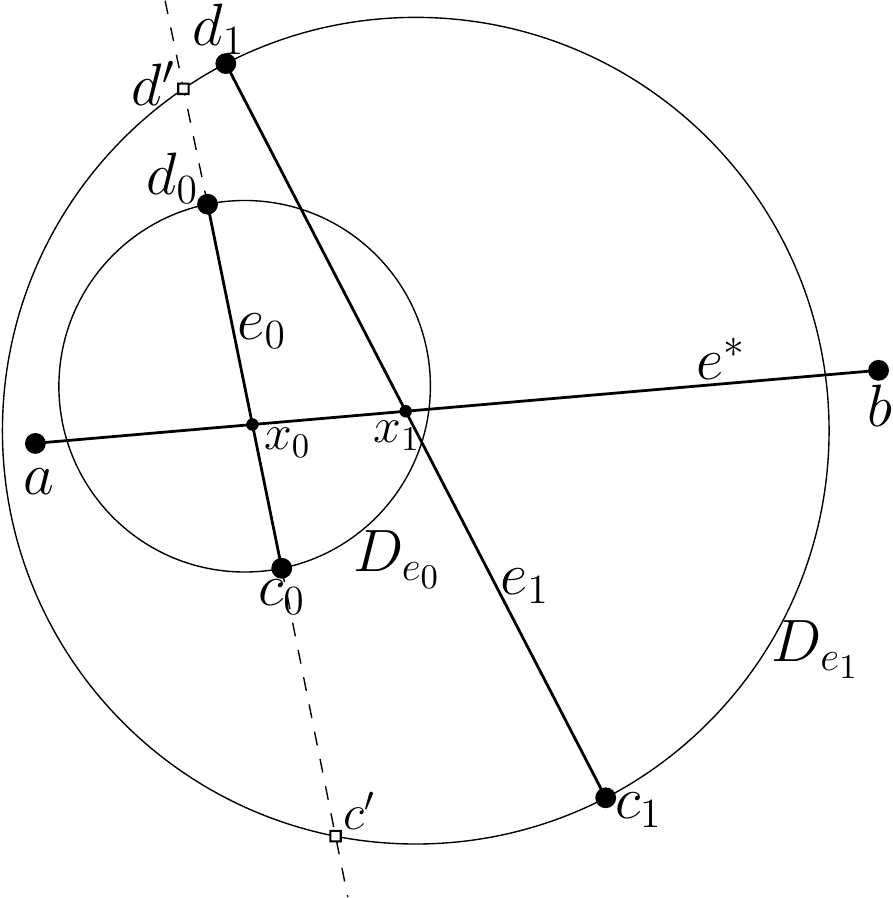}}
\caption{Both endpoints of $e^*$ are outside the disk with $e_i$ as a diameter for every $0 \leq i < k$.}
\label{fig:star_zerosecondbar}
\end{figure}

Assume to the contrary that $a \in D_{e_0}$. We showed in the previous step that there does not exist a point of $V$ inside $D_{e_0}$ that is visible to both $c_0$ and $d_0$. Thus, the point $a$ must be blocked from $c_0$ or $d_0$ by an edge of $F$. Notice, that edges of $F$ cannot intersect neither line segment $\overline{a x_0}$ (refer to the definition of $e_0$) nor $e_0$ (because $F$ is plane and $e_0 \in F$). Therefore, the edge that blocks $a$ from $c_0$ (respectively, $d_0$) must have one endpoint inside the triangle $\triangle (a, c_0, x_0)$ (respectively, $\triangle (a, x_0, d_0 )$). Since both triangles are inside $D_{e_0}$, this endpoint also belongs to $D_{e_0}$ and must not be visible to both $c_0$ and $d_0$. It must be blocked similarly to $a$. Since we have a finite number of points in $V$ there will be a point in $D_{e_0}$ that will \textbf{not} be blocked from $c_0$ and $d_0$. This is a contradiction to $a \in D_{e_0}$.

Consider the edge $e_1$. We want to show that $a \notin D_{e_1}$. We proved that $a \notin D_{e_0}$. Assume to the contrary that $a \in D_{e_1}$ and thus $\angle(c_1, a, d_1) > \pi/2$. Refer to Fig.~\ref{fig:star_secondbar}. It follows, that $x_0 \in D_{e_1}$. We claim that at least one endpoint of $e_0$ belongs to $D_{e_1}$. Assume to the contrary that $c_0$, $d_0 \notin D_{e_1}$. Then the edge $e_0$ intersects the boundary of $D_{e_1}$ in points $c'$ and $d'$. Since $e_0$ and $e_1$ do not intersect, by construction the chord $\overline{c'd'}$ lies between $a$ and $e_1$ and therefore $\angle (c', a, d') > \angle (c_1, a, d_1) > \pi/2$. On the other hand, since $a \notin D_{e_0}$, we have $\angle (c', a, d') < \angle (c_0, a, d_0) < \pi/2$. We have a contradiction and thus at least one endpoint of $e_0$ belongs to $D_{e_1}$. Without loss of generality assume that $c_0 \in D_{e_1}$. We proved in step $1$ that there are no points of $V$ inside $D_{e_1}$ that are visible to both $c_1$ and $d_1$. Therefore, the point $c_0$ must be blocked from $c_1$ or $d_1$ by an edge of $F$. Notice that $\overline{x_0 x_1}$, $e_0$ and $e_1$ cannot be intersected by edges of $F$. It follows that one endpoint of the blocking edge must be inside quadrilateral $\{c_0, c_1, x_1, x_0\}$ or triangle $\triangle(x', x_1, d_1)$, where $x'$ is the intersection point between $\overline{c_0 d_1}$ and $e^*$. Both polygons belong to $D_{e_1}$ and thus the endpoint of the blocking edge also belongs to $D_{e_1}$. This endpoint cannot be visible to both $c_1$ and $d_1$ and thus must be blocked by another edge of $F$. Since the number of points in $V$ is finite there will be a point in $D_{e_1}$ that will be visible to both $c_1$ and $d_1$. This is a contradiction and thus $a \notin D_{e_1}$.

In a similar way we can show that $a \notin D_{e_i}$ for every $e_i$, $0 \leq i < k$. Symmetrically, it is also true that $b \notin D_{e_i}$ for $0 \leq i < k$. Thus we have shown that for every edge $e \in F$ that intersects $e^*$, both endpoints of $e^*$ are outside $D_e$.

\textbf{Step 3:} In the previous step we showed that $a,b \notin D_{e_i}$ for $0 \leq i < k$. Thus $e^*$ intersects the boundary of $D_{e_i}$ at the points $a'$ and $b'$. Refer to Fig.~\ref{fig:star_finalbar}. 

\renewcommand\windowpagestuff{%
\centering
\captionsetup{width=.8\linewidth}
\includegraphics[width=0.62\textwidth]{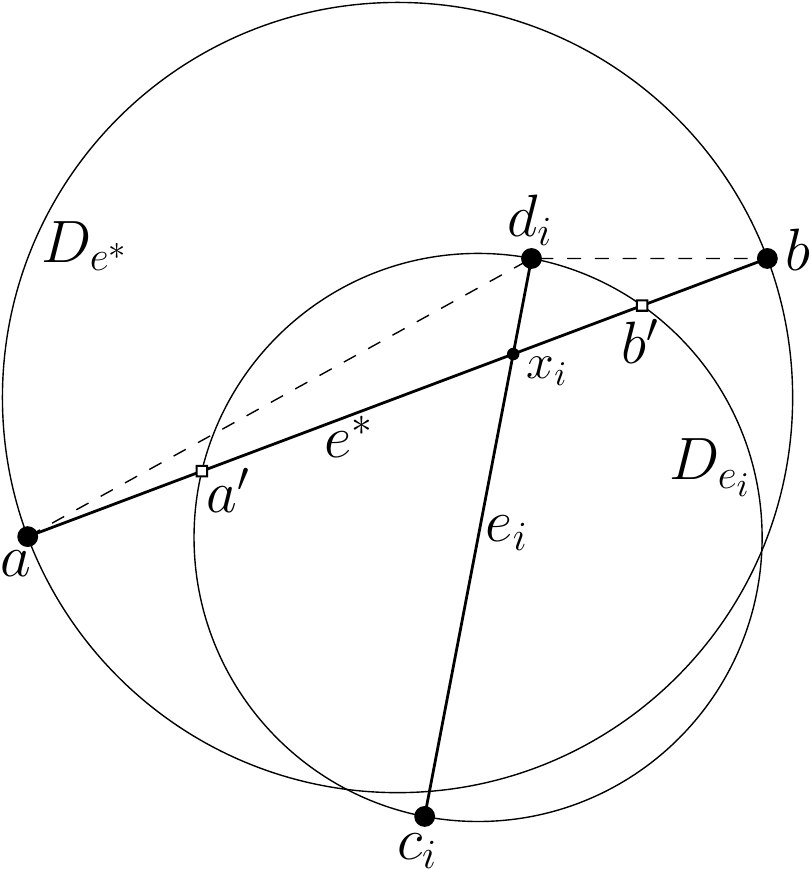}
\captionof{figure}{Intersection between $e^* \in CMST(V,S)$ and an edge of $F$.}
\label{fig:star_finalbar}
}
\opencutright
\begin{cutout}{2}{0.5\linewidth}{0pt}{14}

Consider the chord $\overline{a'b'}$ of $D_{e_i}$. Since $e_i$ and $\overline{a'b'}$ intersect, then $\angle( a', d_i, b')\geq \pi/2$ or $\angle( b',c_i,a')\geq \pi/2$. Without loss of generality assume that $\angle( a', d_i, b')\geq \pi/2$. It follows that $\angle( a, d_i, b) > \pi/2$ and thus $d_i \in D_{e^*}$. Consider the graph $VG(V,S)$. We are given that $e^* \in CMST(V,S)$. We prove in Lemma~\ref{lem:Hierarchy} (refer to Sect.~\ref{sec:Hierarchy}) that $CMST(V,S) \subseteq CGG(V,S)$. Since $e^* \notin S$ then by Def.~\ref{def:locallyGabriel} and~\ref{def:CGG(I)} the disk $D_{e^*}$ does not contain points in $V$ which are visible from both $a$ and $b$. Therefore the point $d_i$ must be blocked from $a$ or $b$ by an edge of $S$. Notice, that edges of $S$ cannot intersect neither $e^*$ (because $e^* \in VG(V,S)$) nor $e_i$ (because $S \subseteq E$ and $F$ is plane). Thus, one endpoint of the edge that blocks $d_i$ from $a$ or $b$ must be inside the triangle $\triangle (a, b, d_i)$. This endpoint belongs to $D_{e^*}$ (since $\triangle (a, b, d_i) \in D_{e^*}$) and therefore must be blocked from $a$ or $b$ by another edge of $S$. We have a finite number of points in $V$ and thus there is a point of $V$ in $D_{e^*}$ that is not blocked from $a$ and $b$. This is a contradiction to $e^* \in CGG(V,S)$.
\end{cutout}

We conclude that an intersection between an edge of $F$ and an edge of $CMST(V,S)$ is not possible. Therefore $e^* \in VG(F)$, meaning $CMST(V, S) \subseteq VG(F)$.
\qed
\end{proof}

\begin{lemma}
\label{lem:CMST=T}
Let $S$ be the output of Algorithm~\ref{alg_CMST} on the input plane forest $F = (V, E)$. We have $F \subseteq CMST(V, S)$. 
\end{lemma}

\begin{proof}
Let $e = \overline{uv}$ be an arbitrary edge of $F$. We have to show that $e \in CMST(V,S)$. 
There are two cases to consider:
\begin{enumerate}
\item $e \in S$: then $e \in CMST(V,S)$ by Def.~\ref{def:CMST(F)}.
\item $e \notin S$: then by Property~\ref{property1} we have $e \in MST(EVG(F))$. Notice, that $e \in CMST(F)$ because $e \in F$. Consider the cut $U$, $V \setminus U$ in $CMST(F)$ such that $u \in U$, $v \in V \setminus U$ (refer to Fig.~\ref{fig:cut2}). 

\begin{figure}[h]
  \begin{minipage}[c]{0.44\textwidth}
    \includegraphics[width=0.94\textwidth]{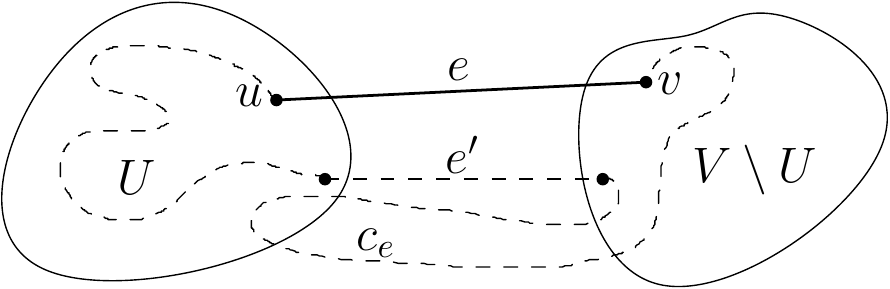}
  \end{minipage}\hfill
  \begin{minipage}[c]{0.56\textwidth}
    \caption{The cut $U$, $V \setminus U$ in $CMST(F)$ s.t. $u \in U$, $v \in V \setminus U$, $e \in F$. The path $c_e \setminus \{e\}$ from $u$ to $v$ in $CMST(V,S)$ is shown in dashed line. By Lemma~\ref{lem:CMSTinVG(T)} none of the edges of this path intersects $e$. Notice, that $e, e' \notin S$ and $e$ is the heaviest edge of the cycle $c_e$, thus $w(e) > w(e')$.}
\label{fig:cut2}
  \end{minipage}
\end{figure}


Notice, that none of the edges of $S$ belong to this cut, otherwise $CMST(F)$ has a cycle (because $S \in F \subseteq CMST(F)$). Assume to the contrary that $e \notin CMST(V,S)$. If we add $e$ to $CMST(V,S)$ we will create a cycle $c_e$, that contains the path $c_e \setminus \{e\}$ from $u$ to $v$ in $CMST(V,S)$. By Lemma~\ref{lem:CMSTinVG(T)} we have $c_e \subseteq VG(F)$. Therefore, there exist an edge $e' \in c_e$ that belongs to the cut $U$, $V \setminus U$ in $VG(F)$. Notice, that $e$ must be the heaviest edge of the cycle $c_e$, therefore, since $e' \notin S$ then $|e| > |e'|$ and thus $w(e) > w(e')$. Two cases are possible:
\begin{enumerate}
\item $e' \in MST(EVG(F))$: Algorithm~\ref{alg_CMST} considers $e'$ in step $4$. Since $e$ and $e'$ belong to the same cut $U$, $V \setminus U$ a cycle in $CMST(F) \cup \{e'\}$ is formed. This cycle contains $e$. Because $|e| > |e'|$ the edge $e$ is added to $S$. This contradicts to $e \notin S$.

\item $e' \notin MST(EVG(F))$: there exist a cycle $c_{e'}$ in the graph $MST(EVG(F)) \cup \{e'\}$ such that $e'$ is the heaviest edge of this cycle. Consider the cut $U'$, $V \setminus U'$ in $CMST(V,S)$ such that one endpoint of $e'$ is in $U'$ and another is in $V \setminus U'$ (refer to Fig.~\ref{fig:cut3}). There exist an edge $\widehat{e}$ of $c_{e'}$ that belongs to the cut $U'$, $V \setminus U'$ in $VG(F)$. Since $e' \notin S$ we can delete it from $CMST(V,S)$. Notice that $CMST(V,S) \setminus \{e'\} \cup \{\widehat{e}\} $ is a tree that contains all the edges of $S$. Moreover, since $|e'| > |\widehat{e}|$ the weight of $CMST(V,S) \setminus \{e'\} \cup \{\widehat{e}\}$ is smaller than the weight of $CMST(V,S)$, which is a contradiction.
\end{enumerate}
\end{enumerate}
We showed a contradiction to the fact that that $e \notin CMST(V,S)$ and thus $F \subseteq CMST(V, S)$. 
\qed 
\end{proof}

\begin{figure}[h]
  \begin{minipage}[c]{0.44\textwidth}
    \includegraphics[width=0.94\textwidth]{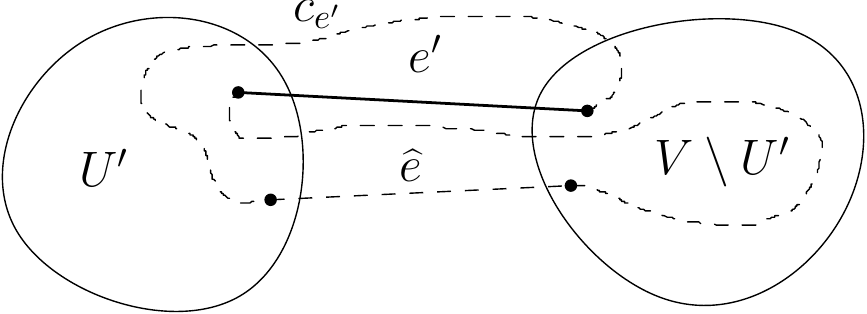}
  \end{minipage}\hfill
  \begin{minipage}[c]{0.56\textwidth}
\caption{The cut $U'$, $V \setminus U'$ in $CMST(V,S)$ s.t. $e' \in CMST(V,S)$ and $e'$ belongs to this cut. Recall that $e' \notin S$. The cycle $c_{e'}$ in $MST(EVG(F)) \cup \{e'\}$ is shown in dashed line. By Lemma~\ref{lem:CMSTinVG(T)} none of the edges of this cycle is intersected by $e'$. Notice, that $e'$ is the heaviest edge of $c_{e'}$.}
\label{fig:cut3}
  \end{minipage}
\end{figure}


\begin{lemma}
\label{lem:minimal&minimum}
Let $S$ be the output of Algorithm~\ref{alg_CMST} on the input plane forest $F = (V, E)$. The set $S$ is minimal and minimum. 
\end{lemma}

\begin{proof}
Let us first prove that the set $S$ is minimal. Assume to the contrary that $S$ is not minimal. Thus there is an edge $e \in S$ such that $F \subseteq CMST(V, S \setminus \{ e \})$. We have to disprove that $F$ is a subgraph of $CMST(V, S \setminus \{ e \})$. Let $T$ be $CMST(V, S \setminus \{ e \})$. Since $e \in S$, the edge $e$ was added to $S$ by Algorithm~\ref{alg_CMST}. Thus, there exists an edge $e' \in MST(EVG(F))$ such that $e' \notin F$ (and thus $e' \notin CMST(F)$); and there is a cycle $c_{e'}$ in $CMST(F) \cup \{ e'\}$ such that $e \in c_{e'}$ and $|e| > |e'|$. Notice, that $T^* = T \setminus \{ e \} \cup \{e' \}$ is a tree. Moreover, $T^*$ is a subgraph of $CVG(V,S \setminus \{ e \})$ whose weight is smaller than the weight of $T$. This is a contradiction to $T$ being $CMST(V, S \setminus \{ e \})$.

Let us now prove that the set $S$ is minimum. Assume to the contrary that there is another set $S^* \subseteq E$ of constraints, such that $F \subseteq CMST(V, S^*)$ and $|S^*| < |S|$. We have to show a contradiction to $F \subseteq CMST(V, S^*)$.

We proved that $S$ is minimal, thus $S \setminus S^* \neq \emptyset$. Let $e = \overline{uv}$ be an edge of $S \setminus S^*$. Notice, that $e \in F$. Consider the cut $U$, $V \setminus U$ in $CMST(F)$ such that $u \in U$, $v \in V \setminus U$ (refer to Fig.~\ref{fig:cut}). Algorithm~\ref{alg_CMST} added $e$ to $S$, thus there exists an edge $e' \in MST(EVG(F))$ such that $e' \notin E$, $e$ and $e'$ belong to the cycle $c_{e'}$ in $CMST(F) \cup \{e'\}$ and $|e| > |e'|$. Notice, that $e'$ belongs to the cut $U$, $V \setminus U$ in $VG(F)$. Moreover, $e' \in CVG(V, S^*)$ because $S^* \subseteq E$ and $e'$ does not intersect edges of $F$.
Since $e$ is not a constraint in $CMST(V,S^*)$, then its weight is not $0$ but equal to the Euclidean distance between its endpoints. Thus, the inequality $w(e) > w(e')$ is true for $CVG(V,S^*)$. 

The edge $e'$ is not intersected by any edge of $F$ (because $e'\in VG(F)$) and thus the graph $G = CMST(V, S^*) \setminus \{e\} \cup \{e'\}$ is a plane tree and $w(G) < w(CMST(V, S^*))$. Thus $G$ is the true $CMST(V, S^*)$. However $F \nsubseteq G$ because $e \notin G$. This is a contradiction.
\qed
\end{proof}


Lemmas~\ref{lem:CMST=T} and~\ref{lem:minimal&minimum} show the correctness of Algorithm~\ref{alg_CMST}. However, we said nothing about our strategy of finding cycles in the graph. With a naive approach step $4$ -- $8$ of the algorithm could be quadratic in~$n$. Also, the size of the visibility graph of $F$ can be quadratic in the size of $V$, leading to the complexity of step $1$ of the algorithm equal to $O (n^2 \log n )$. Our first step to improve the running time is to reduce the size of the graph we construct MST for. Lemma~\ref{lem:MST(EVG(F))inCDT(F)} shows that $MST(EVG(F)) \subseteq CDT(F)$. The same lemma can be used to show that $CMST(F)$ can be constructed from $CDT(F)$. Notice that $CDT(F)$ has size $O(n)$. The running time of steps $1$ and $2$ then becomes $O(n \log n)$. Moreover, if $F$ is a plane tree then the construction of $CDT(F)$ can be performed in $O(n)$ time~\cite{doi:10.1137/S0097539795285916}.

\begin{lemma}
\label{lem:MST(EVG(F))inCDT(F)}
Given a plane forest $F=(V,E)$ we have $MST(EVG(F)) \subseteq CDT(F)$. Notice that the edges of both graphs $EVG(F)$ and $CDT(F)$ are assigned weights equal to the Euclidean distance between the endpoints of corresponding edges. Similarly, $MST(CVG(F)) \subseteq CDT(F)$.
\end{lemma}
\begin{proof}
Let $e = \overline{uv}$ be an arbitrary edge of $MST(EVG(F))$. If $e \in F$ then by Def.~\ref{def:CDT(F)}, $e \in CDT(F)$. Assume that $e \notin F$. Consider the graph $CDT(F)$. We consider the circle with the line segment $\overline{u v}$ as a diameter. Suppose there is a point $p$ inside this circle, that is visible to $u$ and $v$. Refer to Fig.~\ref{fig:MSTinCDT}. Then we have $|\overline{p u}| < |\overline{u v}|$ and $|\overline{p v}| < |\overline{u v}|$ (by $|\overline{p u}|$ we denote the length of the line segment $\overline{p u}$). Consider the cycle $u \rightarrow p \rightarrow v \rightarrow u$. We know that the heaviest edge on a cycle does not belong to any $MST$. Hence, $e \notin MST(EVG(F))$ because $e$ is the heaviest edge of that cycle. This is a contradiction to the assumption that there is a point $p$ inside the circle. Therefore, there is a circle having $u$ and $v$ on its boundary that does not have any vertex of $F$ visible to $u$ and $v$ in its interior. By Def.~\ref{def:locallyDelaunay} the pair of vertices $u$ and $v$ form a locally Delaunay edge, thus $e \in CDT(F)$.

In a similar we can show that $MST(CVG(F)) \subseteq CDT(F)$.
\qed
\end{proof}

We use the Link/Cut Tree of Sleator and Tarjan~\cite{Sleator:1981:DSD:800076.802464} to develop an efficient solution for the forth step (lines $4$ -- $8$) of Algorithm~\ref{alg_CMST}. The complexity of the algorithm becomes $O(n \log n)$.

\subsection{An Efficient Implementation of Algorithm 1}

In this subsection we develop an efficient solution for the step $4$ -- $8$ of Algorithm~\ref{alg_CMST}. We use Link/Cut Tree type of data structure invented by Sleator and Tarjan in 1981~\cite{Sleator:1981:DSD:800076.802464}. It is also referred in literature as a Dynamic Trees data structure. This data structure maintains a collection of rooted vertex-disjoint trees and supports two main operations: \emph{link} (that combines two trees into one by adding an edge) and \emph{cut} (that divides one tree into two by deleting an edge). Each operation requires $O(\log n)$ time. We consider the version of dynamic trees problem that maintains a forest of trees, each of whose edges has a real-valued cost. We are also interested in a fast search for an edge that has the maximal cost among edges on a tree path between a pair of given vertices. We can implement this with the help of original \emph{mincost}(vertex $v$) operation (that returns the vertex $w$ closest to the root of $v$ such that the edge between $w$ and its parent has minimum cost among edges on the tree path from $v$ to the root of $v$). However, for that to work in our case we have to negate weights of our problem. Alternatively, we can create \emph{maxcost}(vertex $v$) operation, by implementing it similarly to \emph{mincost}(vertex $v$). Then we can use unaltered weights. Below is a brief list of operations we use. Refer to~\cite{Sleator:1981:DSD:800076.802464} for a detailed description and implementation.

\begin{itemize}
\item \textbf{parent}(vertex $v$): Returns the parent of $v$, or $null$ if $v$ is a root and thus has no parent.
\item \textbf{root}(vertex $v$): Returns the root of the tree containing $v$. 
\item \textbf{cost}(vertex $v$): Returns the cost of the edge ($v$, \textbf{parent}($v$)). This operation assumes that $v$ is not a tree root. 
\item \textbf{maxcost}(vertex $v$): Returns the vertex $u$ closest to \textbf{root}($v$) such that the edge ($u$, \textbf{parent}($u$)) has maximum cost among edges on the tree path from $v$ to \textbf{root}($v$). This operation assumes that $v$ is not a tree root. 
\item \textbf{link}(vertex $v$, $u$, real $x$): Combines the trees containing $v$ and $u$ by adding the edge ($v$, $u$) of cost $x$, making $u$ the parent of $v$. This operation assumes that $v$ and $u$ are in different trees and $v$ is a tree root. 
\item \textbf{cut}(vertex $v$): Divides the tree containing vertex $v$ into two trees by deleting the edge ($v$, \textbf{parent}($v$)); returns the cost of this edge. This operation assumes that $v$ is not a tree root.
\item \textbf{update edge}(vertex $v$, real $x$): Adds $x$ to the cost of 
the edge ($v$, \textbf{parent}($v$)).
\item \textbf{lca}(vertex $v$, $u$): Returns the lowest common ancestor of $v$ and $u$, or returns \textbf{null} if such an ancestor does not exist.
\end{itemize}

The data structure uses $O(n)$ space and each of the above operations takes $O(\log n)$ time. Algorithm~\ref{alg_CMST_Imp} finds the minimum set $S \subseteq E$ for a given forest $F = (V, E)$ such that $F \subseteq CMST(V,S)$. Algorithm's complexity is $O(n \log n)$.

\begin{algorithm}[H]
\caption{$O(n \log n)$ time implementation of Algorithm~\ref{alg_CMST}}\label{alg_CMST_Imp}
\KwIn{plane forest $F = (V, E)$}
\KwOut{minimum set $S \subseteq E$ of constraints such that $F \subseteq CMST(V,S)$}
\BlankLine
Construct $T' = MST(CDT(F))$\;
Construct $CMST(F) = MST(CDT^\circ (F))$\;
Initialize $S = \emptyset$\;
Construct Link/Cut tree data structure of $CMST(F)$, such that each edge $e = \overline{uv} \in CMST(F)$ gets cost $w(e)$, equal to the Euclidean distance between vertices $u$ and $v$\;
\ForEach{$e' = \overline{v_1v_2} \in T'$ (such that $e' \notin CMST(F)$)}{
	\If{\textbf{lca}($v_1$, $v_2$) $\neq$ \textbf{null}}{
		$u$ = \textbf{lca}($v_1$, $v_2$)\;
		$p$ = \textbf{parent}($u$)\;
		\lIf{$p \neq null$}{$y$ = \textbf{cut}($u$)}
		\If{\textbf{parent}($v_1$) $\neq null$}{
			$v = $ \textbf{maxcost}($v_1)$\;
			$x = $ \textbf{cost}($v$)\;
			\While{$\big(w(e) < x\big)$}{
				\textbf{update edge}($v$, $-x$)\;
				set $S \leftarrow S \cup \{ (v, $  \textbf{parent}($v$)$) \}$\;
				$v = $ \textbf{maxcost}($v_1)$\;
				$x =$ \textbf{cost}($v$)\;
			}		
		}
		\If{\textbf{parent}($v_2$) $\neq null$}{
			$v = $ \textbf{maxcost}($v_2)$\;
			$x = $ \textbf{cost}($v$)\;
			\While{$\big(w(e) < x\big)$}{
				\textbf{update edge}($v$, $-x$)\;
				set $S \leftarrow S \cup \{ (v, $  \textbf{parent}($v$)$) \}$\;
				$v = $ \textbf{maxcost}($v_2)$\;
				$x =$ \textbf{cost}($v$)\;
			}		
		}
		\lIf{$p \neq null$}{\textbf{link}($u$, $p$, $y$)}
	}
}
Set $S \leftarrow S \cap E$\;
\end{algorithm}

If $F \subseteq MST(V)$ then $F \subseteq CMST(V, \emptyset)$. In other words we do not require constraints at all to obtain Constrained MST that will contain $F$. It is interesting to consider the opposite problem. How big can the set of constraints be? Figure~\ref{fig:worst_case} shows the worst-case example, where the set $S$ of constraints contains all the edges of $F$, thus $|S| = n-1$. 

\begin{figure}[h]
  \begin{minipage}[c]{0.5\textwidth}
    \includegraphics[width=0.94\textwidth]{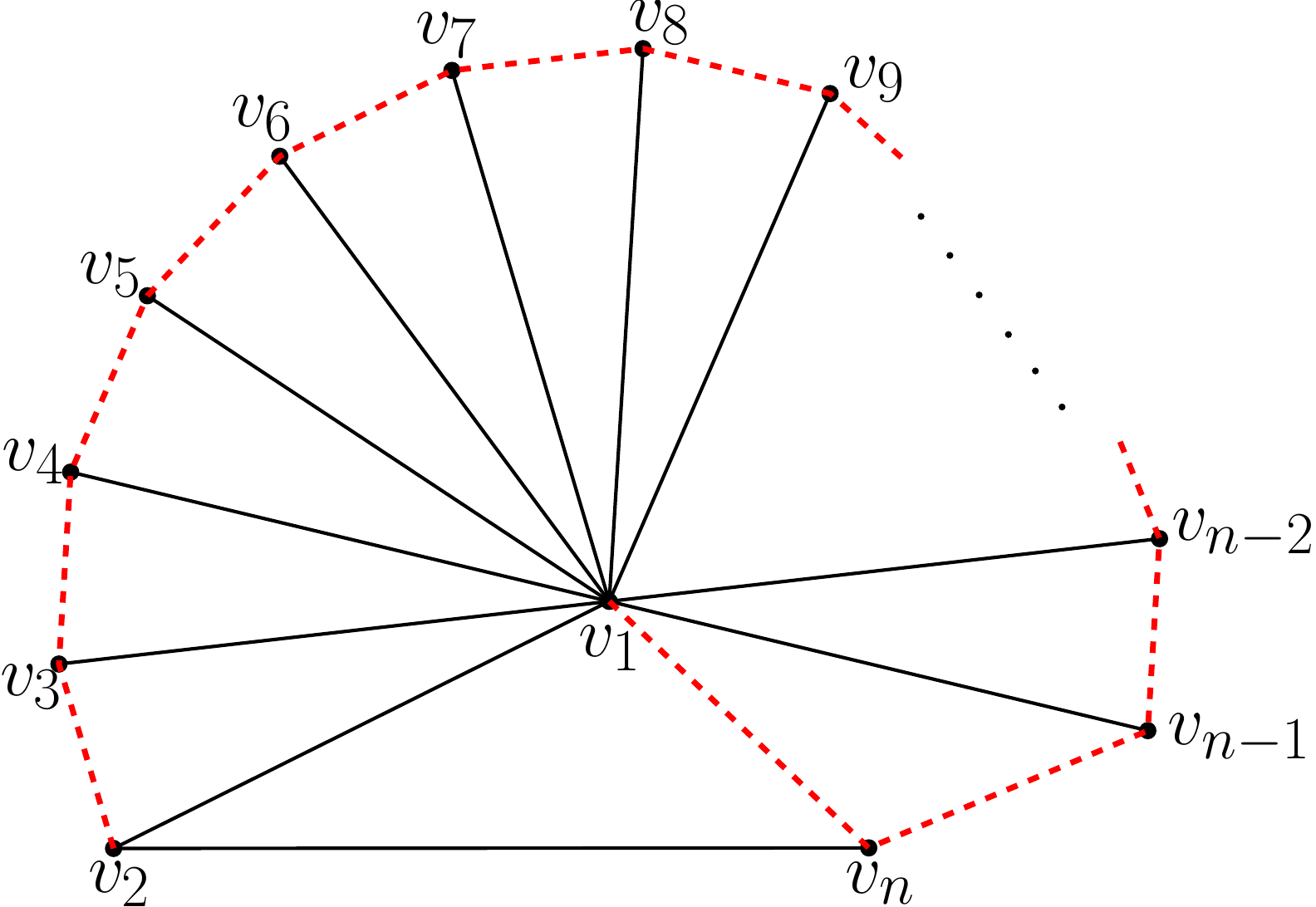}
  \end{minipage}\hfill
  \begin{minipage}[c]{0.5\textwidth}
\caption{Worst case example showing $n-1$ constraints. The input tree $F$ is drawn using solid lines. The $MST$ of the set $\{v_1, v_2, \ldots v_n\}$ is dashed.}
  \label{fig:worst_case}
  \end{minipage}
\end{figure}

\begin{theorem}
Given a plane forest $F = (V, E)$, Algorithm~\ref{alg_CMST_Imp} constructs the minimum set $S \subseteq E$ of constraints such that $F \subseteq CMST(V,S)$. The running time of Algorithm~\ref{alg_CMST_Imp} is $O(n \log n)$, where $n = |V|$.
\end{theorem}

\section{Constrained Gabriel Graph Algorithm}
\label{sec:CGG}

\textbf{Problem 2:} We are given a plane graph $I = (V, E)$ of $|V| = n$ points. Find the minimum set $S \subseteq E$ of edges such that $I \subseteq CGG(V, S)$. In other words, we are interested in the minimum set $S \subseteq E$ of edges such that $CGG(V, S) = CGG(I)$.

Algorithm~\ref{alg_BETA} given in Sect.~\ref{sec:beta-skeleton} can be successfully applied to $CGG$s. However, the algorithm presented in this section is significantly simpler and requires only local information about an edge in question. We can decide whether or not the edge $e \in I$ should be in $S$ by considering at most two triangles adjacent to $e$ in $CDT(I)$.
This can be done in constant time. We exploit the fact that $S$, constructed of all non locally Gabriel edges of $I$, is necessary and sufficient. Consider the following lemma.

\begin{lemma}
\label{lem:S-notLocallyGabriel}
Let $I = (V, E)$ be a plane graph and let $S$ be the set of all the edges of $E$ that are not locally Gabriel. Then $S$ is the minimum set of constraints such that $I \subseteq CGG(V,S)$.
\end{lemma}
\begin{proof}
Let us first show that $I \subseteq CGG(V,S)$. Indeed, $S \subseteq I$ so we can add edges to $S$ to obtain $E$. Thus, we add $E \setminus S$, which, by definition, consists of locally Gabriel edges only. Therefore, there is no edge in $I$ that does not belong to $CGG(V,S)$. So, $S$ is a sufficient set of constraints.

Let us show that $S$ is necessary. Let $S'$ be a subset of edges in $I$ such that $I \subseteq CGG(V,S')$ and $S \setminus S' \neq \emptyset$. Let $e = (u,v)$ be an edge in $S \setminus S'$. By definition, the edge $e$ is not locally Gabriel and thus there exists a point $p \in V$ such that the circle with $e$ as a diameter contains $p$. Since $I \subseteq CGG(V,S')$, then in particular $e \in CGG(\{u, v, p\}, \emptyset )$ which is false since $e$ is not locally Gabriel. Hence such an edge cannot exist and thus $S$ is minimum.
\qed
\end{proof}

The above lemma gives us a tool for constructing $S$. We need to find a proper graph that is both relatively small in size and keeps required information about each edge of $I$ easily accessible. As you may have already guessed, $CDT(I)$ is a good candidate. We show in Lemma~\ref{lem:Hierarchy} that $CGG(I) \subseteq CDT(I)$. Thus if an edge is locally Gabriel in $CGG(I)$ then it is locally Delaunay in $CDT(I)$. Applied to our problem, it means that if an edge of $I$ is not a constraint in $CGG(I)$ then it is not a constraint in $CDT(I)$. The opposite however is not true. 
Refer to Fig.~\ref{fig:not_locally_Gabriel} for an example of an edge that is locally Delaunay, but not locally Gabriel. 
\begin{wrapfigure}{r}{0.42\textwidth}
\vspace{-10pt}
\centering
\includegraphics[width=0.41\textwidth]{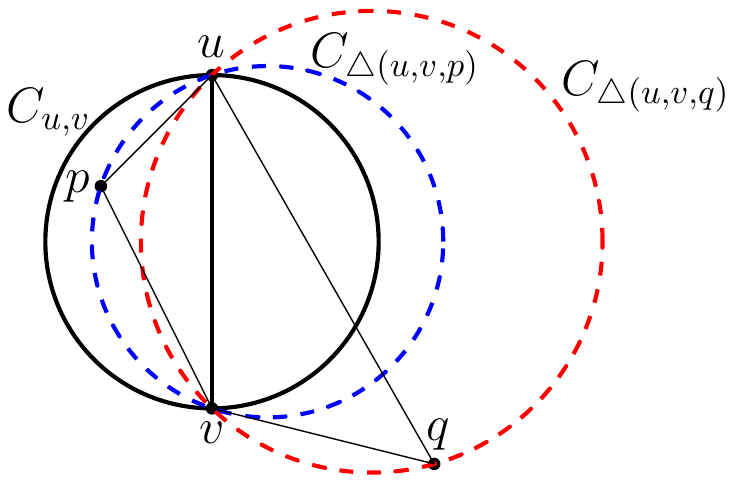}
\caption{The edge $\overline{uv}$ is locally Delaunay, but not locally Gabriel. If $\overline{uv} \in I$ then $\overline{uv} \notin S_{CDT(I)}$ and $\overline{uv} \in S_{CGG(I)}$. If $\overline{uv} \notin I$ then $\overline{uv} \notin CGG(I)$.}
\vspace{-10pt}
\label{fig:not_locally_Gabriel}
\end{wrapfigure}


We prove in Lemma~\ref{lem:constraint_travel_main} that $S_{CGG(I)} \supseteq S_{CDT(I)}$, where $S_G$ denotes the \textbf{minimum} set of constraints of $G$, such that $I \subseteq G$. This means that if some edge of $I$ is a constraint in $CDT(V, S_{CDT(I)})$ then the same edge is also a constraint in $CGG(V, S_{CGG(I)})$. Notice, that $CDT(V, S_{CDT(I)}) = CDT(I)$ and $CGG(V, S_{CGG(I)}) = CGG(I)$. Thus, if $S_{CDT(I)}$ is known then we can speed up our algorithm by initializing $S$ with $S_{CDT(I)}$.

The only question that is left unanswered is how can we tell in constant time if an edge of $I$ is locally Gabriel or not by observing $CDT(I)$. The following lemma 
answers this question.

\begin{lemma}
\label{lem:edge_of_CDT}
Let $I = (V, E)$ be a plane graph and let $\overline{uv}$ be an edge of $I$ such that $\overline{uv}$ is \textbf{not} a constraint edge of $CDT(I)$. Let $\triangle(u,v,p)$ be a triangle in $CDT(I)$, and let $C_{u,v}$ be the circle with diameter $\overline{uv}$. If $p \notin C_{u,v}$ then there is \textbf{no} point $x$ of $V$ on the same side of $\overline{uv}$ as $p$ such that $x \in C_{u,v}$ and $x$ is visible to both $u$ and $v$ in $CGG(I)$. Refer to Fig.~\ref{fig:uv_not_constrained}.
\end{lemma}

\renewcommand\windowpagestuff{%
\centering
\captionsetup{width=.8\linewidth}
\includegraphics[width=0.62\textwidth]{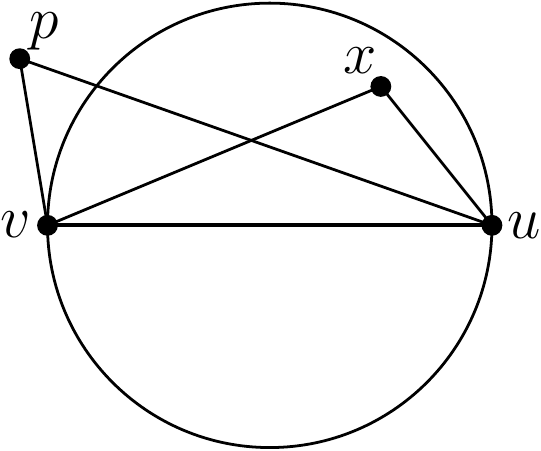}
\captionof{figure}{Proof of Lemma~\ref{lem:edge_of_CDT}.}
\label{fig:uv_not_constrained}
}

\begin{proof} 
\opencutright
\begin{cutout}{2}{0.56\linewidth}{0pt}{10}
Assume to the contrary, that such a point $x$ exists. Since $x$ is visible to both $u$ and $v$ in $CGG(I)$, then since $S_{CGG(I)} \supseteq S_{CDT(I)}$ (refer to Lemma~\ref{lem:constraint_travel_main}) there are no constraints in $CDT(I)$ intersected by the line segments $\overline{ux}$ and $\overline{vx}$. Refer to Fig.~\ref{fig:uv_not_constrained}. Since $p \notin C_{u,v}$, the angle $\angle vpu$ is acute. Because $x \in C_{u,v}$, the angle $\angle vxu$ is obtuse. Consider the circle $C'$ circumscribing the triangle $\triangle(u, v, p)$. Both angles $\angle vpu$ and $\angle vxu$ are supported by the same chord $\overline{uv}$, and since $x$ and $p$ are on the same side of $\overline{uv}$ and $\angle vpu < \angle vxu$, then $x$ is inside $C'$. The edge $\overline{up}$ is not locally Delaunay and since it belongs to $CDT(I)$, then $\overline{up}$ must be a constraint. The line segment $\overline{vx}$ and the constraint $\overline{up}$ intersect, which is a contradiction to the existence of the point~$x$.
\qed
\end{cutout}
\end{proof}

We are ready to present an algorithm for constructing $S$.

\begin{algorithm}[H]
\caption{Minimum set of constraints for constrained Gabriel graphs}\label{alg_CGG}
\KwIn{plane graph $I = (V, E)$}
\KwOut{minimum set $S \subseteq E$ of constraints such that $I \subseteq CGG(V, S)$}
\BlankLine
Construct $CDT(I)$\;
Initialize $S=\emptyset$\;
\ForEach{$\overline{u v} \in E$}{
	Consider the triangle $\triangle (u, v, p)$ and possibly the triangle $\triangle (u, v, q)$ for $q\neq p$ of $CDT(I)$\;
	Consider the circle $C_{u,v}$ with $\overline{u v}$ as a diameter\;
	\If{$p \in C_{u,v}$ or $q \in C_{u,v}$}{set $S \leftarrow S \cup \{ \overline{u v} \}$}
}
\end{algorithm}
 
Notice that $CGG(V,S) \subseteq CDT(I)$, because  $CGG(V, S) = CGG(I) \subseteq CDT(I)$. Lemma~\ref{lem:edge_of_CDT} justifies step $3$ -- $7$ of the algorithm. Notice, that the set $S$ constructed by Algorithm~\ref{alg_CGG}, is the set of all the edges of $I$ that are \textbf{not} locally Gabriel. Thus, by Lemma~\ref{lem:S-notLocallyGabriel}, $S$ is the minimum set of constraints such that $I \subseteq CGG(V,S)$. This concludes the proof of correctness of Algorithm~\ref{alg_CGG}.

The running time of step $1$ of Algorithm~\ref{alg_CGG} is $O(n \log n)$ in the worst case. The step $3$ -- $7$ of the algorithm takes $O(n)$ time. If the input graph $I$ is a simple polygon, a triangulation or a tree then the running time of the first step of the algorithm can be reduced to $O(n)$, leading to a \textbf{linear} total running time for this algorithm.

\begin{theorem}
Given a plane graph $I = (V, E)$, Algorithm~\ref{alg_CGG} constructs the minimum set $S \subseteq E$ of constraints such that $I \subseteq CGG(V, S)$. The running time of Algorithm~\ref{alg_CGG} is $O(n \log n)$, where $n = |V|$.
\end{theorem}

\section{Constrained $\beta$-skeleton Algorithm}
\label{sec:beta-skeleton}

\textbf{Problem 3:} We are given a plane graph $I = (V, E)$ of $|V| = n$ points and $1 \leq \beta \leq 2$. Find the minimum set $S \subseteq E$ of edges such that $I \subseteq CG_\beta(V, S)$. In other words, we are interested in the minimum $S$ such that $CG_\beta(V, S) = CG_\beta(I)$.

For the constrained Gabriel graph, the problem can be solved in a simpler way. We show that $S$, constructed of all non locally Gabriel edges of $I$, is necessary and sufficient. We can decide in constant time whether or not the edge $e \in I$ is locally Gabriel by considering at most two triangles adjacent to $e$ in $CDT(I)$. Refer to Sec.~\ref{sec:CGG}.

Let $u$, $v$ and $p$ be a triple of vertices of $V$. Recall the definition of $U_{u,v}(\beta)$; see Sect.~\ref{sec:basic_definitions} and Fig.~\ref{fig:triangle_is_empty0}. If $p \in U_{u,v}(\beta)$ and $p$ is visible to both $u$ and $v$, then we say that the vertex $p$ \emph{eliminates} line segment $\overline{uv}$. We prove in Lemma~\ref{lem:Hierarchy} (refer to Sect.~\ref{sec:Hierarchy}) that $CRNG(I) = CG_{\beta = 2}(I) \subseteq CG_{1 \leq \beta \leq 2}(I) \subseteq CG_{\beta = 1}(I) = CGG(I) \subseteq CDT(I)$. The following lemmas further explain a relationship between $CG_\beta(I)$ and $CDT(I)$.

\begin{figure}[h]
\centering
\subfigure[]{
		\label{fig:triangle_is_empty0}
		\includegraphics[width=0.305\textwidth]{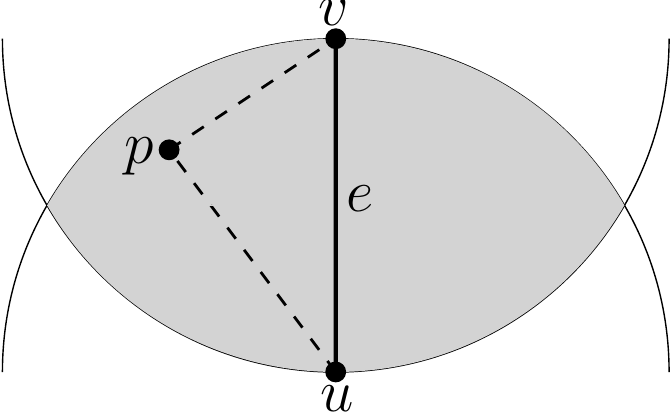}}%
\hspace{0.02\textwidth}
\subfigure[]{
		\label{fig:triangle_is_empty1}
		\includegraphics[width=0.305\textwidth]{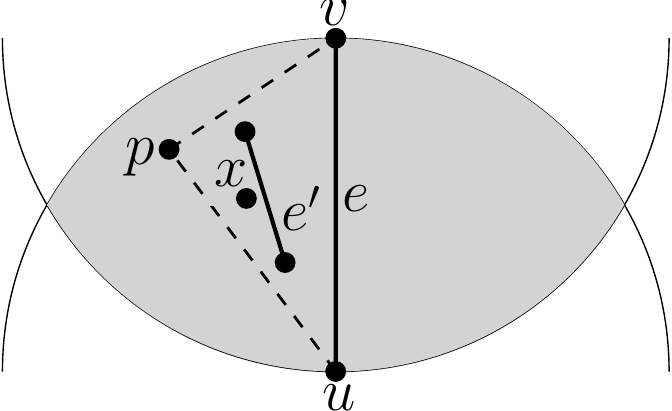}}%
\hspace{0.02\textwidth}
\subfigure[]{
		\label{fig:triangle_is_empty2}
		\includegraphics[width=0.305\textwidth]{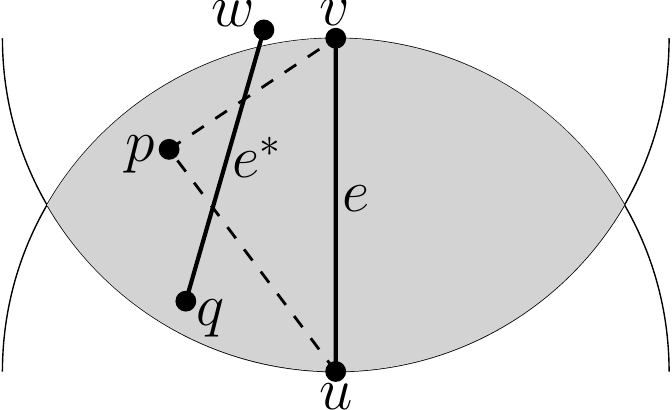}}%
\caption{The neighbourhood $U_{u,v}(\beta=2)$ is highlighted in gray. \textbf{(a)} The vertex $p \in V$ eliminates $\overline{uv}$. \textbf{(b)} If  $p$  is the closest vertex to $\overline{uv} \in CDT(I)$ that eliminates $\overline{uv}$, then $\triangle(u,v,p) \cap V =\emptyset$. \textbf{(c)} If  $p$  is the closest vertex to $\overline{uv} \in CDT(I)$ that eliminates $\overline{uv}$, then $p\in U_{e^*}(\beta)$ for every edge of $e^* \in CDT(I)$ that lies between $p$ and $\overline{uv}$.}
\label{fig:triangle_is_empty}
\end{figure}

\begin{lemma}
\label{lem:beta_alg_empty_triangle}
Given a plane graph $I = (V, E)$ and $1 \leq \beta \leq 2$. Let  $p \in V$  be the closest vertex to the edge $\overline{uv} \in CDT(I)$ that eliminates $\overline{uv}$.
Then $\triangle(u,v,p) \cap V = \emptyset$. Refer to Fig.~\ref{fig:triangle_is_empty}.
\end {lemma}

\begin{proof}
Assume to the contrary that there exist a vertex $x \in V$ such that the triangle $\triangle(u,v,p)$ contains $x$. The vertex $x$ does not eliminate $\overline{uv}$, otherwise it will be a contradiction to $p$ being the closest vertex to $\overline{uv}$ that eliminates $\overline{uv}$. Thus, there exist an edge $e' \in I$ that blocks $x$ from $u$ or $v$, i.e. $e'$ intersects line segment $\overline{ux}$ or $\overline{vx}$. Refer to Fig.~\ref{fig:triangle_is_empty1}. Notice, that $e'$ does not intersect neither $\overline{up}$ nor $\overline{vp}$, otherwise $p$ is not visible to both $u$ and $v$, which contradicts to the fact that $p$ eliminates $\overline{uv}$. Therefore, both endpoints of $e'$ are inside $\triangle(u,v,p)$. Let us look at the endpoint of $e'$ that is closest to $\overline{uv}$ (let us call it $x'$). It is clear that $x'$ is closer to $\overline{uv}$ than $x$. Similarly to $x$, $x'$ must be blocked by an edge of $I$ from either $u$ or $v$. Since we have a finite number of vertices, there will be a vertex inside $\triangle(u,v,p)$ visible to both $u$ and $v$ and closer to $\overline{uv}$ than $p$. This is a contradiction.
\qed
\end{proof}

The above lemma implies that if there exists an edge $e^* \in CDT(I)$ that lies between $p$ and $\overline{uv}$, i.e. $e^*$ intersects the interior of $\triangle(u,v,p)$, then $e^*$ intersects both line segments $\overline{up}$ and $\overline{vp}$. Refer to Fig.~\ref{fig:triangle_is_empty2}. Notice, that $e^*$ cannot intersect $\overline{uv}$, since $e^*, \overline{uv} \in CDT(I)$.

\begin{lemma}
\label{lem:beta_alg_triangle_intersection}
Given a plane graph $I = (V, E)$ and $1 \leq \beta \leq 2$. Let  $p \in V$  be the closest vertex to the edge $\overline{uv} \in CDT(I)$ that eliminates $\overline{uv}$.
Let $e^* = \overline{qw}$ be an edge of $CDT(I)$ that intersects $\triangle(u,v,p)$. Then $p \in U_{q,w}(\beta)$. Ref. to Fig.~\ref{fig:triangle_is_empty2}.
\end {lemma}

\begin{proof}
According to Lemma~\ref{lem:beta_alg_empty_triangle}, the edge $e^*$ intersects $\overline{up}$ and $\overline{vp}$. Refer to Fig.~\ref{fig:triangle_is_empty2}. Notice, that $e^* \notin I$, otherwise $p$ is not visible to $u$ and $v$. The endpoints of $e^*$ may or may not belong to $U_{u,v}(\beta)$. Notice, that if an endpoint of $e^*$ belongs to $U_{u,v}(\beta)$ then it must be either further away from $\overline{uv}$ than $p$ or it must be blocked by some edge of $I$ from being visible to $u$ or $v$.
Observe, that $\angle(q,p,w) > \angle(u,p,v)$. Thus, since $p \in U_{u,v}(\beta)$ then $p \in U_{q,w}(\beta)$.
\qed
\end{proof}

To solve our main problem we will use two geometric structures: \emph{elimination path} and \emph{elimination forest}, introduced by Jaromczyk and Kowaluk in~\cite{Jaromczyk:1987:NRN:41958.41983}. The elimination path for a vertex $p$ (starting from an adjacent triangle $\triangle(p,u,v)\in CDT(I)$) is an ordered list of edges, such that $p \in U_{e}(\beta)$ for each edge $e$ of this list. In the work~\cite{Jaromczyk:1987:NRN:41958.41983} an edge $e$ belongs to the elimination path induced by some point $p$ only if $e$ is eliminated by $p$. In our problem this is not the case. The point $p$ eliminates $e$ if and only if $p \in U_{e}(\beta)$ and $p$ is visible to both endpoints of $e$. We show how to adapt the original elimination forest to our problem later in this section. 

The elimination path is defined by the following construction (refer to Algorithm~\ref{alg_elimPATH}). Notice, that the construction of the elimination path introduced by Jaromczyk and Kowaluk in~\cite{Jaromczyk:1987:NRN:41958.41983} does not contain step $4$.

\begin{algorithm}[H]
\caption{Construction of the elimination path}\label{alg_elimPATH}
\KwIn{vertex $p \in V$, $\triangle(p,u,v)\in CDT(I)$ and $1 \leq \beta \leq 2$}
\KwOut{$path(p, \overline{uv})$ - ``elimination path of $p$ towards $\overline{uv}$"}
\BlankLine
$e := \overline{uv}$, $path(p, \overline{uv}) := \emptyset$\;
\While{$p \in U_{e}(\beta)$}{
	$path(p, \overline{uv}) := append (path(p, \overline{uv}), e)$\;
	\lIf{$e$ is \textbf{not} locally Delaunay}{terminate \textbf{while} loop}
	Take triangle $\triangle t \in CDT(I)$ adjacent to $e$, such that $\triangle t$ and $p$ are separated by the line passing through $e$\;
	\lIf{$p \in U_{e'}(\beta)$ such that $e' \in \triangle t$}{set $e := e'$;}
}
\end{algorithm}

Jaromczyk and Kowaluk show that a vertex cannot belong to the $\beta$-neighbour- hood $U(\beta)$ of more than two edges for a particular triangle. Thus elimination paths do not split. Moreover, they also show that if two elimination paths have a common edge $e$ and they both reached $e$ via the same triangle, then starting from this edge one of the two paths is completely included in another one. This property is very important - it guarantees that the elimination forest can be constructed in linear time. Refer to~\cite{Jaromczyk:1987:NRN:41958.41983} and~\cite{Jaromczyk:1989} for further details. 

Since we are dealing with $CDT$, the elimination paths defined via the original construction may split at non-locally Delaunay edges. To overcome this problem we terminate the propagation of the elimination path after a non-locally Delaunay edge is encountered and added to the path. Thus, the elimination forest for our problem can also be constructed in linear time. It is shown in Lemma $2.3$ in~\cite{Jaromczyk:1987:NRN:41958.41983} that if two points eliminate a common edge of a triangle in $DT$ (such that both points are external to this triangle) then the two points can eliminate at most one of the remaining edges of this triangle. Similarly, we can show that if two points of $V$ eliminate a common locally Delaunay edge $e$ of an external triangle in $CDT$ then they can eliminate at most one of the remaining edges of the triangle. It is due to the fact that there exists a circle that contains the endpoints of $e$ such that if any vertex $v$ of $V$ is in the interior of the circle then it cannot be ``seen" from at least one of the endpoints of $e$. It means that the point $v$ does not eliminate $e$,-- the elimination path of $v$ terminated at non-locally Delaunay edge that obstructed visibility between $v$ and one of the endpoints of~$e$.

Lemmas~\ref{lem:beta_alg_empty_triangle} and~\ref{lem:beta_alg_triangle_intersection} show that no important information will be lost as a result of ``shorter" elimination paths. Every non-locally Delaunay edge of $CDT(I)$ is a constraint and thus belongs to $I$. Edges of $I$ obstruct visibility. Let  $p \in V$  be the closest vertex to the edge $\overline{uv} \in CDT(I)$ that eliminates $\overline{uv}$. Assume to the contrary that $\overline{uv}$ is not on the elimination path from $p$ because the path terminated at non-locally Delaunay edge $e^*$ before the path could reach $\overline{uv}$. By Lemma~\ref{lem:beta_alg_empty_triangle} the edge $e^*$ intersects both line segments $\overline{pu}$ and $\overline{pv}$. Refer to Fig.~\ref{fig:triangle_is_empty2}. Since $e^* \in I$, neither $u$ nor $v$ are visible to the point $p$. This contradicts the fact that $p$ eliminates $\overline{uv}$. Lemma~\ref{lem:beta_alg_triangle_elimination_path} further shows that if some edge $e$ of $I$ must be a constraint in $CG_{\beta}(I)$ then it will belong to at least one elimination path, and in particular, to the path induced by the closest point to $e$ that eliminates $e$.

\begin{lemma}
\label{lem:beta_alg_triangle_elimination_path}
Given a plane graph $I = (V, E)$ and $1 \leq \beta \leq 2$. Let  $p \in V$  be the closest vertex to the edge $\overline{uv} \in CDT(I)$ that eliminates $\overline{uv}$.
Then $\overline{uv}$ belongs to the elimination path induced by  $p$.
\end{lemma}

\begin{proof}
Let $E'$ be the set of all the edges of $CDT(I)$ that lie between $p$ and $\overline{uv}$. By Lemma~\ref{lem:beta_alg_empty_triangle} each of those edges intersects both line segments $\overline{pu}$ and $\overline{pv}$. We proved in Lemma~\ref{lem:beta_alg_triangle_intersection} that $p \in U_{e^*}(\beta)$ for every edge $e^* \in E'$. Since $p$ eliminates $\overline{uv}$, none of the edges of $E'$ belong to $I$ and thus every edge of $E'$ is locally Delaunay. Let $e'$ be the closest edge to $p$ among all the edges of $E'$. It follows that every edge of $E'$ belongs to the elimination path induced by $p$ towards edge $e'$: $path(p, e')$. Therefore, the edge $\overline{uv}$ is also appended to this elimination path at a certain point.
\qed
\end{proof}

Each elimination path starts with a special node (we call it a leaf) that carries information about the vertex that induced the current elimination path. A node that corresponds to the last edge of a particular elimination path also carries information about the vertex that started this path. The elimination forest is build from bottom (leaves) to top (roots).

The elimination forest (let us call it $ElF$) gives us a lot of information, but we still do not know how to deal with visibility. The elimination path induced by point $p$ can contain locally Delaunay edges of $I$ that may obstruct visibility between $p$ and other edges that are further on the path. We want to identify edges that not only belong to the elimination path of some vertex $p$ but whose both endpoints are also visible to $p$. Observe, that visibility can only be obstructed by edges of $I$. Let us contract all the nodes of the $ElF$ that correspond to edges not in $I$. If a particular path is completely contracted, we delete its corresponding leaf as well. Now the $ElF$ contains only nodes of edges that belong to $I$ together with leaves, that identify elimination paths, that originally had at least one edge of $I$. The correctness of our approach is supported by the following lemma.

\begin{lemma}
\label{lem:beta_alg_triangle_elimination_path_contracted}
Given a plane graph $I = (V, E)$, $1 \leq \beta \leq 2$ and a contracted $ElF$ of $CDT(I)$. There exists a vertex of $V$ that eliminates $\overline{uv} \in I$ \textbf{if and only if} the node $n_{\overline{uv}}$ of the contracted $ElF$ has a leaf attached to it.
\end{lemma}

\begin{proof}
We start by proving the first direction of the lemma. Let  $p \in V$  be the closest vertex to the edge $\overline{uv} \in CDT(I)$ that eliminates $\overline{uv}$. By Lemma~\ref{lem:beta_alg_triangle_elimination_path}, $\overline{uv}$ belongs to the elimination path induced by $p$. If $\overline{uv}$ is the first edge on this elimination path or every edge on this path between $p$ and $\overline{uv}$ is not in $I$, then $n_{\overline{uv}}$ has a leaf (that represents this path) in contracted $ElF$. Assume to the contrary that $n_{\overline{uv}}$ has no immediate leaf, meaning that elimination path from $p$ to $\overline{uv}$ contains an edge of $I$. By Lemma~\ref{lem:beta_alg_triangle_intersection} this edge intersects both $\overline{pu}$ and $\overline{pv}$, which contradicts to $p$ being visible to $u$ and $v$.

Let us prove the opposite direction. Let $n_{\overline{uv}}$ be a node of contracted $ElF$ that has a leaf attached to it. We have to prove that there exist a vertex that eliminates $\overline{uv}$. Notice also, that if $ElF$ does not contain node $n_{\overline{uv}}$ then $\overline{uv}$ does not belong to any elimination path. Among all the possible leaves attached to  $n_{\overline{uv}}$ let us choose the one that corresponds to the vertex $p$ closest to $\overline{uv}$. According to the construction of $ElF$, elimination path induced by $p$ (let us call it $path_p$) contains $\overline{uv}$. Notice that elimination path represents a list of triangles of $CDT(I)$, such that each pair of consecutive edges of the path belongs to the same triangle. If $\triangle(p,u,v) \in CDT(I)$ then $p$ eliminates $\overline{uv}$ and we are done. Assume that $\triangle(p,u,v) \notin CDT(I)$. All the vertices of every triangle of $CDT(I)$ on $path_p$ are not inside $\triangle(p,u,v)$. Otherwise, the path induced by such a vertex will overlay with $path_p$, which contradicts to $p$ being the closest to $\overline{uv}$ among other leaves attached to $n_{\overline{uv}}$. Thus all the edges between $p$ and $\overline{uv}$ on $path_p$ intersect both $\overline{pu}$ and $\overline{pv}$. Moreover, none of those edges belong to $I$, otherwise, in contracted $ElF$ $p$ will not be a leaf of $n_{\overline{uv}}$. It follows that in $CDT(I)$ there does not exist an edge of $I$ that intersects $\overline{pu}$ or $\overline{pv}$. Therefore, $p$ is visible to both $u$ and $v$ and thus $p$ eliminates $\overline{uv}$.
\qed
\end{proof}

We are ready to present an algorithm that finds the minimum set $S \subseteq E$ of edges such that $I \subseteq CG_\beta(V, S)$ for constrained $\beta$-skeletons ($1~\leq~\beta~\leq~2$): 

\begin{algorithm}[H]
\caption{Minimum set of constraints for constrained $\beta$-skeletons}\label{alg_BETA}
\KwIn{plane graph $I = (V, E)$ and $1 \leq \beta \leq 2$}
\KwOut{minimum set $S \subseteq E$ of constraints such that $I \subseteq CG_\beta(V, S)$}
\BlankLine
Construct $CDT(I)$\;
Initialize $S=\emptyset$\;
Construct Elimination Forest ($ElF$) of $CDT(I)$\;
\ForEach{$e \notin E$}{
	contract the node that corresponds to $e$ in $ElF$\;
	\If{a particular path is about to be contracted in full}{delete its corresponding leaf from $ElF$}
}
\ForEach{node $n_e$ (that corresponds to edge $e$) of Contracted $ElF$}{
	\If{$n_e$ has an immediate leaf attached to it}{set $S \leftarrow S \cup \{e \}$;}
}
\end{algorithm}

Let us discuss the correctness of Algorithm~\ref{alg_BETA}. Let $S$ be the output of the algorithm on the input plane graph $I = (V, E)$. Notice that the following is true: $I \subseteq CG_\beta(V, S)$. Every edge of $E$ that belongs to $S$ also belongs to $CG_\beta(V, S)$ by definition. According to the algorithm, every edge of $E\setminus S$ does not have a leaf attached to a corresponding node in contracted $ElF$. By Lemma~\ref{lem:beta_alg_triangle_elimination_path_contracted} none of those edges has a vertex that eliminates it.

\begin{lemma}
\label{lem:beta_alg_minimality}
Let $S$ be the output of Algorithm~\ref{alg_BETA} on the input plane graph $I = (V, E)$. The set $S$ is minimum.
\end{lemma}

\begin{proof}
Assume to the contrary that there is another set $S'\subseteq E$, such that $I \subseteq CG_\beta(V, S')$ and $|S| > |S'|$. Let $e'$ be an edge of $S \setminus S'$. Notice, that $e' \in I$. Algorithm~\ref{alg_BETA} added $e'$ to $S$, thus there exists a leaf in the contracted $ElF$ attached to $n_{e'}$. By Lemma~\ref{lem:beta_alg_triangle_elimination_path_contracted} there exists a vertex of $V$ that eliminates $e'$. Thus $e' \notin CG_\beta(V, S')$. This contradicts to $I \subseteq CG_\beta(V, S')$.
\qed
\end{proof}

The running time of Algorithm~\ref{alg_BETA} depends on the complexity of the first step. Steps $2$ -- $10$ can be performed in $O(n)$ time. In the worst case the construction of $CDT(I)$ can take $O(n \log n)$ time. But for different types of input graph $I$ this time can be reduced. If $I$ is a tree or a polygon, then $CDT(I)$ can be constructed in $O(n)$ time. If $I$ is a triangulation,  then $I = CDT(I)$ and thus the first step is accomplished in $O(1)$ time.

\begin{theorem}
Given a plane graph $I = (V, E)$, Algorithm~\ref{alg_BETA} constructs the minimum set $S \subseteq E$ of constraints such that $I \subseteq CG_\beta(V, S)$. The running time of Algorithm~\ref{alg_BETA} is $O(n \log n)$, where $n = |V|$.
\end{theorem}

\section{Hierarchy}
\label{sec:Hierarchy}
Proximity graphs form a nested hierarchy, a version of which was established~in~\cite{DBLP:journals/pr/Toussaint80}:

\begin{wrapfigure}{r}{0.4\textwidth}
\centering
\includegraphics[width=0.39\textwidth]{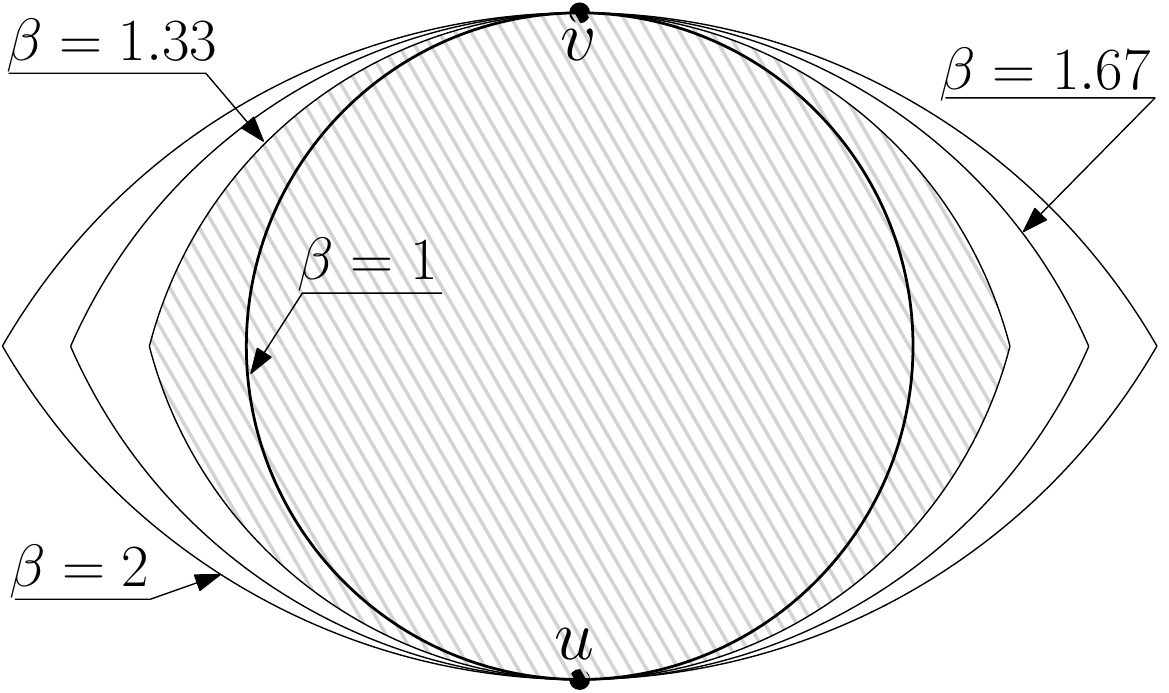}
\caption{Lune-based neighbourhoods for $1 \leq \beta \leq 2$, where $U_{u,v}(\beta = 1.33)$ is dashed.}
\vspace{-35pt}
\label{fig:beta_eye}
\end{wrapfigure}

\begin{theorem}[Hierarchy]
\label{theo:Toussaint_hierarchy}
In any $L_p$ metric, for a fixed set $V$ of points and $1 \leq \beta \leq 2$, the following is true: $MST \subseteq RNG \subseteq G_\beta \subseteq GG  \subseteq DT$.
\end{theorem}

We show that proximity graphs preserve the above hierarchy in the constraint setting (refer to Lemma~\ref{lem:Hierarchy}). We also show that the minimum set of constraints required to reconstruct a given plane graph (as a part of each of those neighbourhood graphs) form an inverse hierarchy (refer to Lemma~\ref{lem:constraint_travel_main}).

\begin{lemma}
\label{lem:Hierarchy}
Given a plane forest $F = (V, E)$ and $1 \leq \beta \leq 2$,
$CMST(F) \subseteq CRNG(F) \subseteq CG_\beta(F) \subseteq CGG(F) \subseteq CDT(F)$. 
Given a plane graph $I = (V, E)$ and $1 \leq \beta \leq 2$, $CRNG(I) \subseteq CG_\beta(I) \subseteq CGG(I) \subseteq CDT(I)$.
\end{lemma}

\begin{proof}
Let us first prove the second statement of the lemma. Assume that we are given a plane graph $I = (V, E)$. We start by proving that $CGG(I) \subseteq CDT(I)$. Let $\overline{uv}$ be an edge of $CGG(I)$. If $\overline{uv} \in E$ then $\overline{uv} \in CDT(I)$ by definition. If $\overline{uv} \notin E$ then $\overline{uv}$ is locally Gabriel, meaning that the circle $C_{u,v}$ with $\overline{uv}$ as a diameter is empty of points of $V$ visible to both $u$ and $v$. By Delaunay criterion $\overline{uv}$ is a Delaunay edge and thus $\overline{uv} \in CDT(I)$. 

Let us prove, that $CG_{\beta_2}(I) \subseteq CG_{\beta_1}(I)$ for $1 \leq \beta_1 \leq \beta_2 \leq 2$. Let $e = (u,v)$ be an arbitrary edge of $CG_{\beta_2}(I)$. If $e \in I$ then by definition $e \in CG_{\beta_1}(I)$, since $e \in I \subseteq CG_{\beta_1}(I)$. If $e \notin I$ then $e$ is not a constraint in $CG_{\beta_2}(I)$ and thus $U_{u,v}(\beta_2)$ does not contain points in $V$ visible to both $u$ and $v$. By construction, $U_{u,v}(\beta_1) \subseteq U_{u,v}(\beta_2)$, refer to Fig.~\ref{fig:beta_eye}. Thus, $U_{u,v}(\beta_1)$ is also empty of points of $V$ visible to both $u$ and $v$. It follows that $e \in CG_{\beta_1}(I)$. We proved the following:
$CRNG(I) = CG_{\beta = 2}(I) \subseteq CG_{1 \leq \beta \leq 2}(I) \subseteq CG_{\beta = 1}(I) = CGG(I)$. We are left to prove the first statement of the lemma. Assume that we are given a plane forest $F = (V, E)$. Since $F$ is a plane graph, the second statement of the lemma is true for $F$, so we only need to prove that $CMST(F) \subseteq CRNG(F)$. Let $e = (u,v)$ be an edge of $CMST(F)$. If $e \in F$ then by definition $e \in CRNG(F)$, since $e \in F \subseteq CRNG(F)$. If $e \notin F$ then $e$ is not a constraint in $CMST(F)$. Assume to the contrary that $e \notin CRNG(F)$, meaning that the lune $L_{u,v}$ is not ``empty". It implies that there is a point $p \in L_{u,v}$ visible to both $u$ and $v$. Refer to Fig.~\ref{fig:triangle_is_empty0}. By construction $|pu| < |vu|$ and $|pv| < |vu|$. Removal of the edge $e$ from $CMST(F)$ disconnects the tree. Assume, without loss of generality, that $v$ and $p$ are in the same connected component of $CMST(F) \setminus \{e\}$. The graph $CMST(F) \setminus \{e\} \cup \{\overline{up}\}$ is a tree. Moreover, $w(CMST(F)) > w(CMST(F) \setminus \{e\} \cup \{\overline{up}\})$. This is a contradiction to $CMST(F)$ being a constrained minimum spanning tree of $F$ and thus $e \in CRNG(F)$.
\qed
\end{proof}

\begin{lemma}
\label{lem:constraint_travel_main}
Let $S_G$ denote the the minimum set of constraints of $G$. Given a plane graph $I = (V, E)$ and $1 \leq \beta \leq 2$, $S_{CRNG(I)} \supseteq  S_{CG_\beta(I)} \supseteq S_{CGG(I)} \supseteq S_{CDT(I)}$. Given a plane forest $F = (V, E)$ and $1 \leq \beta \leq 2$, $S_{CMST(F)} \supseteq S_{CRNG(F)} \supseteq S_{CG_\beta(F)} \supseteq S_{CGG(F)} \supseteq S_{CDT(F)}$.
\end{lemma}

\begin{proof}
First we prove that $S_{CDT(I)} \subseteq S_{CGG(I)}$. Let $e = (u,v)$ be a constraint in $CDT(I)$, namely, $e \in S_{CDT(I)}$. It follows, that $e \in I$ and $e$ is not locally Delaunay and thus there exist a pair of triangles in $CDT(I)$: $\triangle(u,v,p)$ and $\triangle(u,v,q)$, such that $C_{\triangle(u,v,p)}$ contains $q$ or $C_{\triangle(u,v,q)}$ contains $p$, where $C_\triangle$ is a circle circumscribing specified triangle. For example, refer to Fig.~\ref{fig:constraints_travel_1}, where $C_{\triangle(u,v,p)}$ contains $q$. 

\begin{figure}[h]
  \begin{minipage}[c]{0.32\textwidth}
    \includegraphics[width=0.96\textwidth]{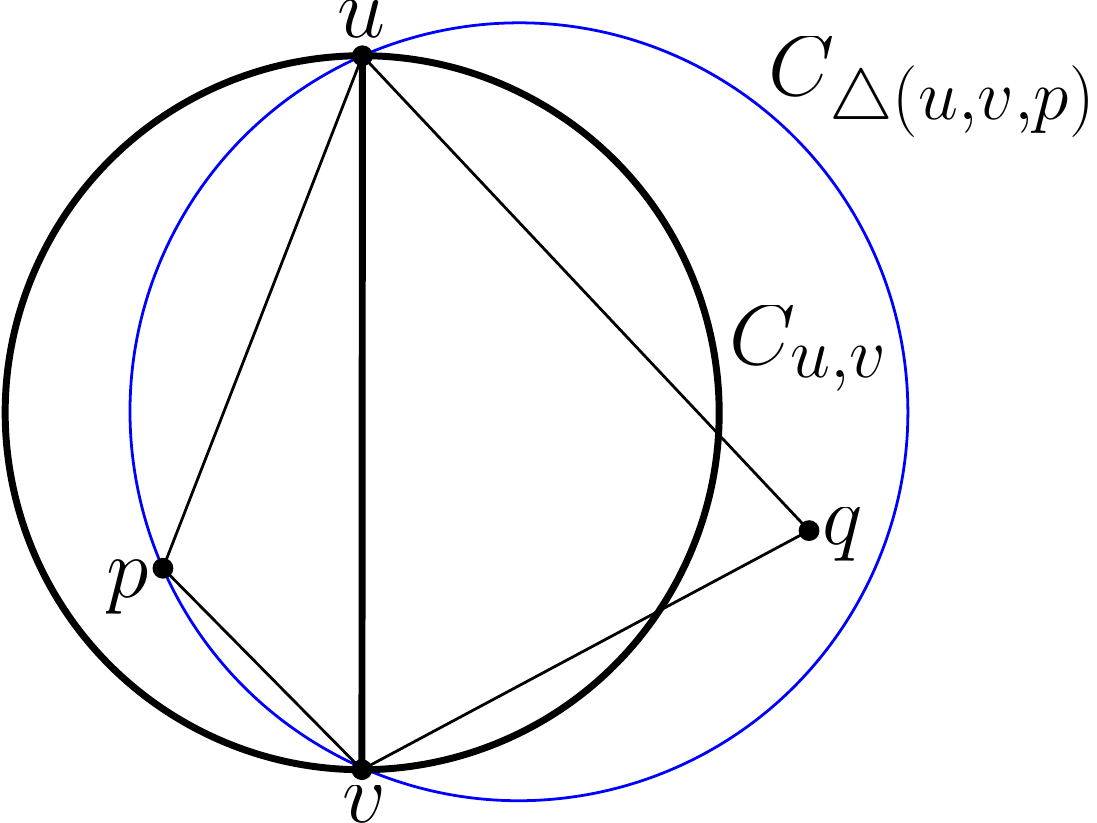}
  \end{minipage}\hfill
  \begin{minipage}[c]{0.6\textwidth}
\caption{$\triangle(u,v,p)$ and $\triangle(u,v,q)$ are triangles of $CDT(I)$. The circle $C_{\triangle(u,v,p)}$ contains $q$ - the edge $\overline{uv}$ is not locally Delaunay.}
\label{fig:constraints_travel_1}
  \end{minipage}
\end{figure}

Consider the quadrilateral $\{ u,p,v,q\}$. At least one of the angles $\angle vpu$ or $\angle uqv$ is obtuse. The following is true for the apex of the obtuse angle:
\begin{enumerate}
 \item It is contained in the circle $C_{u,v}$ with the edge $e$ as a diameter.
 \item It is visible to both $u$ and $v$ in $CDT(I)$, since together with $u$ and $v$ it creates a triangle that belongs to $CDT(I)$. It is also visible to both $u$ and $v$ in $CGG(I)$, because $CGG(I) \subseteq CDT(I)$, refer to Lemma~\ref{lem:Hierarchy}. 
 \end{enumerate} 
The above makes the edge $e$ not locally Gabriel. Since $e \in I$, it must belong to $CGG(I)$, and thus $e \in S_{CGG(I)}$. Therefore $S_{CDT(I)} \subseteq S_{CGG(I)}$.

Notice that $S_{CG_{\beta_1}(I)} \subseteq S_{CG_{\beta_2}(I)}$ for $1 \leq \beta_1 \leq \beta_2 \leq 2$ since $U_{u,v}(\beta_1) \subseteq U_{u,v}(\beta_2)$ for every pair of vertices $u,v \in I$ (refer to Fig.~\ref{fig:beta_eye}). This completes the proof of the first statement of the lemma, which also holds if $I$ is a forest $F$. Thus, we only left to prove that, given a plane forest $F = (V, E)$, $S_{CRNG(F)} \subseteq S_{CMST(F)}$. Let $e = (u,v)$ be an edge in $S_{CRNG(F)}$, meaning that $e \in F$ and $e$ is a constraint in $CRNG(F)$. Thus, there exists a point $p$ visible to both $u$ and $v$ in $CRNG(F)$, such that $p \in L_{u,v}$. By construction $|pu| < |vu|$ and $|pv| < |vu|$. Refer to Fig.~\ref{fig:triangle_is_empty0}. Since $e \in F$ we have $e \in CMST(F)$. Assume to the contrary, that $e \in CMST(F)$ but \textbf{not} a constraint in $CMST(F)$. Removal of $e$ from $CMST(F)$ disconnects the tree. Assume, without loss of generality, that $v$ and $p$ are in the same connected component of $CMST(F) \setminus \{e\}$. The graph $CMST(F) \setminus \{e\} \cup \{\overline{up}\}$ is a tree and its weight is smaller then the weight of $CMST(F)$. This is a contradiction to $CMST(F)$ being a constrained minimum spanning tree of $F$ and thus $e$ must be assigned weight $0$ and be a constraint in $CMST(F)$.
\qed
\end{proof}

%
%

\bibliography{CMST-llncs}

\end{document}